\documentclass{article}

\usepackage[nonatbib, preprint]{cig}

\usepackage[utf8]{inputenc}
\usepackage[T1]{fontenc}

\usepackage{amsmath, amssymb, amsfonts, mathtools, bm}
\usepackage{nicefrac}
\usepackage{bbm}
\usepackage{siunitx}
\usepackage{pifont}
\usepackage{algorithm}
\usepackage{algpseudocode}
\usepackage{amsthm}
\usepackage{calc}

\usepackage{xcolor}
\definecolor{darkblue}{rgb}{0, 0, 0.85}
\definecolor{lightgreen}{rgb}{.85,1,.85}
\definecolor{lightred}{rgb}{1,.85,.85}
\definecolor{lightblue}{rgb}{.85,.85,1}
\definecolor{pink}{HTML}{EB346F}

\usepackage{hyperref}
\hypersetup{
    colorlinks = true,
    citecolor = darkblue,
    linkcolor = black
}

\usepackage{graphicx}
\usepackage{tabularx, threeparttable, tabularray}
\UseTblrLibrary{booktabs}
\usepackage{booktabs}
\usepackage{multirow}
\usepackage{wrapfig}
\usepackage{arydshln}
\usepackage{hhline}
\usepackage{caption}
\usepackage{diagbox}
\usepackage{cite}
\usepackage{microtype}
\usepackage{setspace}
\usepackage{lipsum}
\usepackage{soul}

\usepackage{tikz}

\theoremstyle{plain}

\newtheorem{theorem}{Theorem}
\newtheorem{lemma}{Lemma}

\newtheorem{assumption}{Assumption}

\newtheorem*{suptheorem}{Theorem}

\newcommand{\hlgreen}[1]{{\sethlcolor{lightgreen}\hl{#1}}}

\newcommand{\hlblue}[1]{{\sethlcolor{lightblue}\hl{#1}}}

\DeclarePairedDelimiterX{\infdivx}[2]{(}{)}{#1\;\delimsize\|\;#2}
\newcommand{\infdiv}{D_{\mathsf{KL}}\infdivx}
\def\defn{\,\coloneqq\,}

\renewcommand{\vec}[1]{\bm{#1}} 

\newcommand{\I}{{\vec{I}}}  
\newcommand{\X}{{\vec{X}}}  
  
\newcommand{\Z}{{\vec{Z}}}  
\newcommand{\Q}{{\vec{Q}}}  
\newcommand{\A}{{\vec{A}}}

\def\x{\vec{x}}  
\def\y{\vec{y}}  
\def\z{\vec{z}}  
\def\e{\vec{e}}  
\def\s{\vec{s}}  
\def\n{\vec{n}}  
  
\def\zhat{\widehat{\vec{z}}}

\def\d{\mathrm{d}}

\def\T{\mathcal{T}}
\def\E{\mathbb{E}}
\def\R{\mathbb{R}}
\def\P{\mathbb{P}}
\def\Q{\mathbb{Q}}

\def\Qhat{\widehat{\Q}}
\def\qhat{\widehat{q}}

\title{Measurement Score-Based Diffusion Model}

\vspace{5em}
\author{%
\normalsize Chicago~Y.~Park \quad Shirin Shoushtari \quad Hongyu An \quad Ulugbek S.~Kamilov\\[0.7em]
\small \textnormal{Washington University in St.\ Louis}\\[0.5em]
\footnotesize \texttt{\{chicago, s.shirin, hongyuan, kamilov\}@wustl.edu}
}

\begin{document}

\maketitle

\begin{abstract}
    Diffusion models are widely used in applications ranging from image generation to inverse problems. However, training diffusion models typically requires clean ground-truth images, which are unavailable in many applications. We introduce the \emph{Measurement Score-based diffusion Model (MSM)}, a novel framework that learns \emph{partial} measurement scores using \emph{only} noisy and subsampled measurements. MSM models the distribution of full measurements as an expectation over partial scores induced by randomized subsampling. To make the MSM representation computationally efficient, we also develop a stochastic sampling algorithm that generates full images by using a randomly selected subset of partial scores at each step. We additionally propose a new posterior sampling method for solving inverse problems that reconstructs images using these partial scores. We provide a theoretical analysis that bounds the Kullback–Leibler divergence between the distributions induced by full and stochastic sampling, establishing the accuracy of the proposed algorithm. We demonstrate the effectiveness of MSM on natural images and multi-coil MRI, showing that it can generate high-quality images and solve inverse problems---all without access to clean training data. Code is available at     \href{https://github.com/wustl-cig/MSM}{\textcolor{blue}{\texttt{https://github.com/wustl-cig/MSM}}}.
\end{abstract}

\section{Introduction}

Score-based diffusion models are powerful generative methods that enable sampling from high-dimensional distributions by learning the score function---%
the gradient of the log-density---%
from training data.
They achieve state-of-the-art performance in generating natural images~\cite{dhariwal2021beat}, medical images~\cite{khader2023ddpm3dmedical}, and more, as reviewed in~\cite{chung2024reviewofdiffusion}. Beyond generation, diffusion models can be adapted for conditional sampling given noisy measurements, making them effective for solving inverse problems. However, training requires a large set of clean data, which is often costly or difficult to obtain, such as from hardware limits for high-resolution images or long MRI scan times that cause patient discomfort.

To overcome the limitations of requiring clean data, recent approaches have explored training diffusion models using subsampled~\cite{daras2024ambientdiffusion}, noisy~\cite{xiang2023ddm2, aali2023surescore, daras2024ambienttweedie}, or jointly subsampled and noisy observations~\cite{kawar2023gsure}. These methods aim to approximate the score function of the clean image distribution using degraded measurements. However, these approaches overlook that measurements themselves can be useful, as measurements may either preserve exact, uncorrupted portions of the full measurements or contain only noise in those portions.
To date, no prior work has attempted to directly learn \emph{partial measurement scores} from noisy, subsampled data as a means of modeling the full measurement distribution---which can, in turn, be mapped to the clean image domain.

We introduce the \emph{Measurement Score-based diffusion Model (MSM)}, which learns denoising score functions restricted to observable regions of noisy and subsampled measurements. By taking an expectation over these partial scores---induced by randomized subsampling---MSM effectively models the prior and posterior distribution of full measurements. In the following sections, we present our method, provide theoretical justification, and validate its effectiveness in both unconditional generation and inverse problem solving on natural images and multi-coil MRI data.

\section{Background}
\label{sec:background}

\subsection{Score-Based Diffusion Models}
\label{sec:background_diffusion}

Score-based diffusion models~\cite{song2019generative, ho_NEURIPS2020_ddpm, song2021sde, park2024randomwalks} learn the score function using neural networks.
Tweedie’s formula~\cite{efron2011tweedie} relates the score function to the minimum mean square error (MMSE) denoiser, allowing the score to be estimated using only noisy inputs and their corresponding denoised outputs. Learning score function is performed across varying noise levels, by considering noisy images \(\x_{t} = \x + \sigma_t \n\), where \(\x\) is a clean image, \(\n \sim \mathcal{N}(\vec{0}, \vec{I})\), and \(\sigma_t\) is the noise level at the timestep~\(t\).

Given a denoiser \( \mathsf{D}_\theta \) trained to minimize the mean squared error (MSE), Tweedie’s formula approximates the score function as
$\nabla \log p_{\sigma_t}(\x_t) = \left(\mathsf{D}_{\theta}(\x_t) - \x_t\right)/\sigma_t^2$.
This relationship allows denoisers to serve as practical estimators of the score function at varying noise levels, providing the gradient necessary to drive sampling via reverse-time stochastic processes~\cite{robbins1956empirical, miyasawa1961empirical, vincent2011connection}.

Sampling then proceeds through a sequence of random walks~\cite{park2024randomwalks, ho_NEURIPS2020_ddpm, song2021sde} as
\begin{equation}
\label{eq:randomworks}
    \x_{t-1} = \x_{t} + \tau_t \nabla \log p_{\sigma_t}(\x_t) + \sqrt{2\tau_t\mathcal{T}_t}\n, \quad t = T, T-1, \ldots, 1,
\end{equation}
where $\sigma_t$, $\tau_t$, and $\mathcal{T}_t$ denote the noise-level, step-size, and temperature parameters.
These parameters can be derived from theoretical frameworks~\cite{ho_NEURIPS2020_ddpm, song2021sde} or tuned empirically~\cite{park2024randomwalks}, and the initial sampling iterate \(\x_T\) is drawn from a standard Gaussian to be consistent with the training input of the denoiser.

\subsection{Training Diffusion Models without Clean Data}
\label{sec:background_train_diff}

Training diffusion models to learn the score of clean images typically requires access to high-quality, clean data. However, in many applications, data is often subsampled, noisy, or both.

\textbf{Training with noiseless but subsampled measurements.} Ambient diffusion~\cite{daras2024ambientdiffusion, aali2025ambientposterior} is a recent method for training diffusion models from subsampled measurements by applying an additional subsampling operation during training.
At each step, the model receives a further subsampled and noise-perturbed input and learns to reconstruct the original subsampled measurement.
This procedure jointly encourages denoising and inpainting, guiding the model to approximate the conditional expectation of the clean image given a noisy, partially observed input.

\textbf{Training with noisy but fully-sampled measurements.} SURE-score~\cite{aali2023surescore} enables training score-based diffusion models without access to clean data by leveraging Stein’s unbiased risk estimator (SURE)~\cite{stein1981sureloss}.
SURE-score trains the diffusion model using only noisy measurements, combining two loss functions: a SURE-based loss for denoising the measurements and a denoising loss for the diffusion noise added on top of the denoised estimate.
Another approach~\cite{daras2024ambienttweedie} considers two regimes based on the relationship between measurement noise and diffusion noise.
When measurement noise is relatively smaller, the method estimates the clean image using Tweedie's formula after predicting the noisy measurement from the diffusion iterate.
Conversely, if measurement noise exceeds diffusion noise, the model is trained with a consistency loss~\cite{daras2024consistencyloss}, which encourages stable denoising outputs across nearby timesteps by enforcing that predictions remain consistent along the model’s reverse trajectory.

\textbf{Training with noisy and subsampled measurements.} GSURE diffusion~\cite{kawar2023gsure} trains diffusion models using only noisy, subsampled data by adapting the Generalized SURE loss~\cite{eldar2008generalizedSUREgsure} to the diffusion setting. It reformulates the training objective as a projected loss computable without clean images.
While this objective function is theoretically shown to be equivalent to the supervised diffusion loss under the assumption that the sampling mask and the denoising error are independent, GSURE diffusion has two limitations: it does not support multi-coil MRI, and it requires the minimum diffusion noise level \(\sigma_0\) to match the measurement noise level \(\rho\).
The latter can severely degrade sampling performance when \(\rho\) exceeds typical value of $\sigma_0$ (e.g., $\sigma_0 = 0.01$ in \cite{ho_NEURIPS2020_ddpm, dhariwal2021beat}).
Other recent methods~\cite{bai2024emdiff1, bai2025emdiff2} address the same setting by alternating between reconstructing clean images using diffusion priors pretrained on limited clean data and refining the model to learn from noisy, subsampled measurements.

\subsection{Imaging Inverse Problems}
\label{sec:background_inv_problems}

Inverse problems aim to recover an unknown image \(\x \in \mathbb{R}^p\) from noisy, undersampled measurements \(\y \in \mathbb{R}^m\), modeled as 
\begin{equation}
\label{eq:inverse_problem}
\y = \vec{A}\x + \vec{e}, 
\end{equation}
where \( \vec{A} \in \mathbb{R}^{m \times p} \) is a known forward operator and \( \vec{e} \sim \mathcal{N}(\vec{0}, \eta \vec{I}) \) is Gaussian noise with variance \( \eta \).

When clean training data is unavailable, \emph{self-supervised learning} is often used for training neural networks directly on degraded measurements, without the need for clean ground-truth data~\cite{chen2021equivariant, chen2022robustequivariant, Hu2024spicer, yaman2020ssdu, millard2023ssdupartitioned}. For example, SSDU~\cite{yaman2020ssdu, millard2023ssdupartitioned} is a widely-used approach for training end-to-end restoration networks using distinct subsets of the measurements. In the context of score-based diffusion models, Ambient Diffusion Posterior Sampling (A-DPS)~\cite{aali2025ambientposterior} replacing the clean-image-trained model in the popular DPS method~\cite{chung2023dps} with an Ambient diffusion model trained on subsampled measurements.

\begin{figure*}[t]
\begin{center}
\includegraphics[width=1.\textwidth]{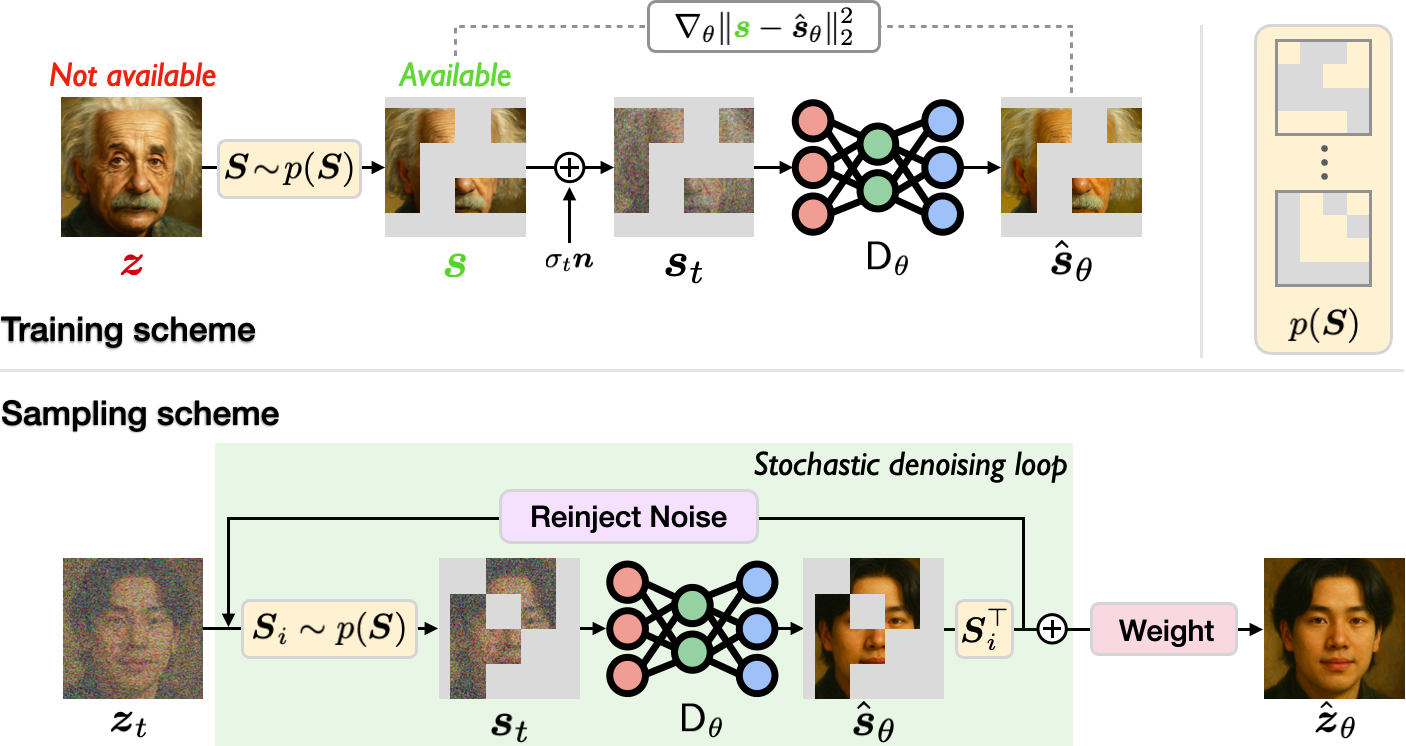}
\end{center}
\caption{Illustration of the \emph{Measurement Score-based diffusion Model (MSM)} for training and sampling using subsampled data. \textbf{Training:} MSM is trained solely on degraded measurements. Diffusion noise is added to these measurements, and the model learns to denoise them. \textbf{Sampling:} At each diffusion step, MSM randomly subsamples the current full-measurement iterate, denoises the resulting partial measurement, and aggregates multiple outputs. A weighting vector compensates for overlapping contributions across partial measurements. See Figure~\ref{fig:illustration_mri} for the MRI-specific version.}

\vspace{-.3cm}
\label{fig:illustration}
\end{figure*}

\section{Measurement Score-Based Diffusion Model}
\label{sec:method}

We present the training and sampling procedures for MSM. A key feature of MSM is that it operates solely on subsampled measurements during training. This enables learning partial measurement scores without access to clean ground-truth data and naturally extends self-supervised denoising to the challenging setting of noisy, subsampled measurements. We also introduce a conditional sampling algorithm that uses the pretrained MSM to solve inverse problems.

\subsection{Learning Partial Measurement Scores}
\label{subsec:learning_measurement_scores}

We first consider the setting where MSM has access only to partial, noiseless observations of an unknown fully-sampled measurement $\z \in \R^n$
\[
    \s = \vec{S}\z \in \R^m,
\]
where \( \vec{S} \in \{0,1\}^{m \times n} \), with $m < n$, is a subsampling mask drawn from distribution \( p(\vec{S}) \). We assume that in the absence of noise, the fully-sampled measurement $\z$ uniquely determines the underlying image $\x$. The definition of \( \z \) depends on the application: for inpainting we set \( \z = \x \), where $\x \in \R^p$ is the clean image; for MRI, we set \( \z = \vec{F}\vec{C}\x \), where $\vec{F}$ is the Fourier transform and $\vec{C}$ is coil-sensitivity operator~\cite{fessler2020optimizationMRIreview}. More broadly, we suppose that fully-sampled measurements are of form $\z = \vec{T}\x$, where $\vec{T} \in \mathbb{R}^{n \times p}$ ($p \geq n$) is an invertible transformation. It is important to note that MSM is designed to operate solely on partial measurements $\s = \vec{S}\z$, with $\vec{S} \sim p(\vec{S})$, without access to any fully-sampled $\z$ or clean image $\x$.

We define a forward diffusion process that adds Gaussian noise to the subsampled measurement $\s$
\[
    \s_t = \s + \sigma_t\n, \quad \n \sim \mathcal{N}(\vec{0}, \vec{I}),
\]
where \( t \in \{1, \dots, T\} \), \( \s_0 = \s \), and \( \s_T \) approaches a known distribution such as a standard Gaussian. At each step \( t \), the diffusion model receives \( \s_t \) as input and returns a denoised estimate of the subsampled measurement $\hat{\s}_{\theta}\in\mathbb{R}^{m}$ as
\begin{equation}\label{eq:denoised_sub_measurement}
\hat{\s}_{\theta}(\s_t \,;\, \sigma_t) = \mathsf{D}_{\theta}(\s_t \,;\, \sigma_t),
\end{equation}
where \( \mathsf{D}_{\theta} \) denotes the noise-conditioned (i.e., time-conditioned) diffusion model.
The model is trained by minimizing the mean squared error (MSE) loss between the predicted and true subsampled measurements:
\[
\mathcal{L}(\theta) = \mathbb{E}_{\s_t, t} \left[\| \s - \hat{\s}_{\theta}(\s_t \,; \sigma_t) \|_2^2 \right].
\]

Once trained, the partial measurement score function can be approximated using Tweedie’s formula~\cite{efron2011tweedie}, which estimates the gradient of the log-probability of the measurement iterate:
\begin{equation} \label{eq:measurement_score}
     \mathsf{S}_{\theta}(\s_t\,;\,\sigma_t)= \frac{1}{\sigma^2_t}(\hat{\s}_{\theta}(\s_t \,;\, \sigma_t) - \s_t),
\end{equation}
where \( \mathsf{S}_{\theta} \in \mathbb{R}^{m}\) denotes the learned partial measurement score. Illustrations for training are provided in Figure~\ref{fig:illustration} and Figure~\ref{fig:illustration_mri}.

\subsection{Unconditional Sampling with MSM} \label{subsec:measurementdiff_prior_sampling}

Implementing a diffusion model on the fully-sampled measurement requires access to the score function \( \nabla \log p_{\sigma_t}(\z_t) \), where \( \z_t \) denotes the noisy version of the fully-sampled measurement \( \z \).
Instead, we train our model to approximate the partial measurement score $\nabla \log p_{\sigma_t}(\s_t)$.

The goal of MSM sampling is to generate a fully-sampled measurement $\z$ given the partial measurement scores in \eqref{eq:measurement_score} for $\vec{S} \sim p(\vec{S})$. To that end, we define the \emph{MSM score} as the expectation over all possible partial scores
\begin{equation} \label{eq:defined_score}
\nabla \log q_{\sigma_t}(\vec{z}_t) \defn \vec{W} \, \mathbb{E}_{\vec{S}\sim p(\vec{S})} \left[ \vec{S}^\top \nabla \log p_{\sigma_t}(\vec{s}_t) \Big|_{\vec{s}_t = \vec{S} \vec{z}_t} \right]
\end{equation}
where each subsampling operator \( \vec{S} \in \mathbb{R}^{m \times n} \) is drawn from distribution \( p(\vec{S}) \), and \( \vec{W} \in \mathbb{R}^n \) is a weighting vector that compensates for overlapping contributions across sampling masks. It is defined as the reciprocal of the expected total coverage:
\begin{equation}
\label{eq:weightvector_expectation}
    \vec{W} \defn \left[ \max \left( \mathbb{E}_{\vec{S} \sim p(\vec{S})} \left[ \mathrm{diag}(\vec{S}^\top \vec{S}) \right],\, 1 \right) \right]^{-1},
\end{equation}
where the maximum is applied elementwise to avoid division by zero in regions not covered by any subsampled measurement.

To efficiently approximate the expectation in~\eqref{eq:defined_score}, we propose a stochastic sampling algorithm that uses a randomly selected subset of partial scores.
Specifically, we stochastically sample \( w \) sampling masks \( \vec{S}^{(i)} \sim p(\vec{S}) \) for \( i = 1, \dots, w \), where \( \vec{S}^{(i)} \in \mathbb{R}^{m_i \times n} \) denotes a subsampling operator $i$.
The corresponding subsampled measurements is obtained as \( \s_t^{(i)} = \vec{S}^{(i)} \z_t \). 
An unbiased estimator of the MSM score is then defined as
\begin{equation} \label{eq:ensemble}
    \nabla \log \widehat{q}_{\sigma_t}(\vec{z}_t) \defn \vec{W} \left[\frac{1}{w} \sum_{i=1}^{w} \vec{S}^{(i)\top} \nabla \log p_{\sigma_t}(\s_t^{(i)})\Big|_{\vec{s}_t^{(i)} = \vec{S}^{(i)} \vec{z}_t}\right],
\end{equation}
where each transpose operator \( \vec{S}^{(i)\top} \in \mathbb{R}^{n \times m_i} \) maps the partial score from the subsampled measurement space back to the fully-sampled measurement space. Note that the reweighting vector \( \vec{W} \) in \eqref{eq:weightvector_expectation} can be empirically estimated as: $\vec{W} \leftarrow \left[\max \left( \sum\nolimits_{i=1}^{w} \mathrm{diag}\left(\vec{S}^{(i)\top} \vec{S}^{(i)} \right),\, 1 \right)\right]^{-1}.$

The estimate of the fully-sampled measurement can be obtained by using the MMSE estimates for subsampled measurements as
\begin{equation} \label{eq:stochastic_denoising}
    \hat{\z}_{\theta} = \vec{W}  \sum_{i=1}^{w} \vec{S}^{(i)\top} \hat{\s}_{\theta}(\s^{(i)}_t\,;\, \sigma_t),
\end{equation}
where $\hat{\s}_{\theta}$ is defined in~\eqref{eq:denoised_sub_measurement}. 

The full sampling procedure is outlined in Algorithm~\ref{alg:kspace_diffusion_sampling}, with illustrative examples shown in Figure~\ref{fig:illustration} and Figure~\ref{fig:illustration_mri}.

\begin{algorithm}[t]
\setstretch{1.5}
\caption{Measurement Score-Based Sampling}
\begin{algorithmic}[1]
\Require $T$, $p(\vec{S}), \{\sigma_t\}^{T}_{t=1}$
\State Initialize $\z_T \sim \mathcal{N}(\vec{0}, \vec{I})$, $\hat{\z}_{\theta} \leftarrow \vec{0}$
\For{$t = T$ \textbf{to} $1$}
    \For{$i = 1$ \textbf{to} $w$}
        \begin{tikzpicture}[remember picture,overlay]
        \node[xshift=3.25cm,yshift=-0.73cm] at (0,0){%
        \includegraphics[width=\textwidth]{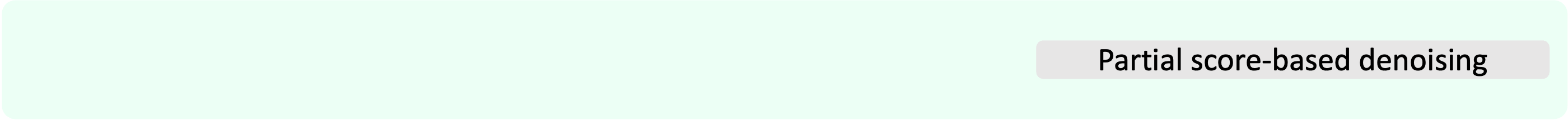}};
        \end{tikzpicture}
        \State $\vec{S}^{(i)} \sim p(\vec{S})$, \; $\s_t^{(i)} \leftarrow \vec{S}^{(i)} \z_t$
        \State $\hat{\s}_{\theta}^{(i)} \leftarrow \s_t^{(i)} + \sigma_t^2 \, \mathsf{S}_{\theta}(\s_t^{(i)}; \sigma_t)$
        \begin{tikzpicture}[remember picture,overlay]
        \node[xshift=1.12cm,yshift=-0.72cm] at (0,0){%
        \includegraphics[width=\textwidth]{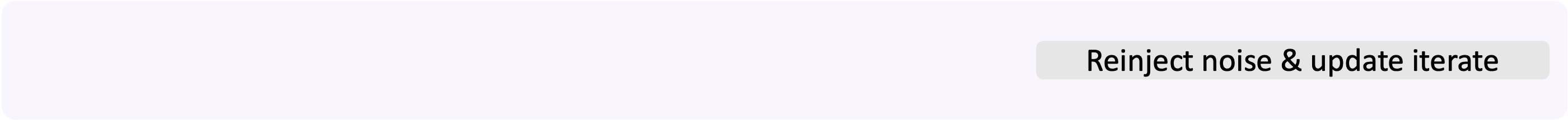}};
        \end{tikzpicture}
        \State $\s_t^{(i)} \sim p(\s_{t}^{(i)} \mid \hat{\s}^{(i)}_{\theta})$
        \State $\z_t \leftarrow \vec{S}^{(i)\top} \s_t^{(i)} + (\vec{I}-\vec{S}^{(i)\top}\vec{S}^{(i)}) \z_t$
    \EndFor
    \begin{tikzpicture}[remember picture,overlay]
    \node[xshift=4.74cm,yshift=-0.75cm] at (0,0){%
    \includegraphics[width=\textwidth]{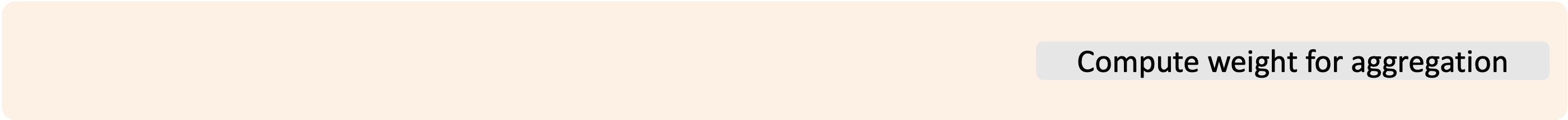}};
    \end{tikzpicture}
    \State $\vec{C} \leftarrow \sum_{i=1}^{w} \mathrm{diag}(\vec{S}^{(i)\top} \vec{S}^{(i)})$ 
    \State $\vec{W} \leftarrow [\max(\vec{C}, 1)]^{-1}$
    \begin{tikzpicture}[remember picture,overlay]
    \node[xshift=2.60cm,yshift=-0.51cm] at (0,0){%
    \includegraphics[width=\textwidth]{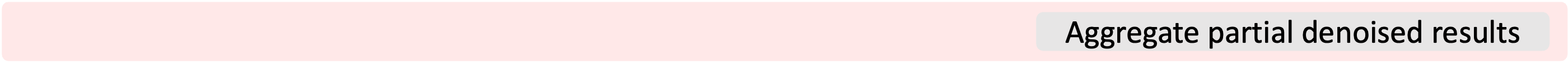}};
    \end{tikzpicture}
    \State $\hat{\z}_{\theta} \leftarrow \vec{W} \sum_{i=1}^{w} \vec{S}^{(i)\top} \hat{\s}_{\theta}^{(i)} + \mathbbm{1}_{\vec{C}=0} \cdot \hat{\z}_{\theta}$ 
    \State $\z_{t-1} \sim p(\z_{t-1} \mid \hat{\z}_{\theta})$
\EndFor
\State \textbf{return} $\z_0$
\end{algorithmic}
\label{alg:kspace_diffusion_sampling}
\end{algorithm}

\subsection{Posterior Sampling with MSM} \label{subsec:measurementdiff_posterior_sampling}

We extend MSM to sample from the posterior distribution for solving inverse problems of the form~\eqref{eq:inverse_problem}.
We consider measurement operators $\A$ of the form $\A = \vec{HT}$, where \( \vec{H} \in \mathbb{R}^{m \times n} \) is the degradation operator---%
such as downsampling, blurring, box inpainting, or random projection---%
and $\vec{T} \in \mathbb{R}^{n \times p}$ is an invertible transformation introduced in Section~\ref{subsec:learning_measurement_scores}.
This allows us to express the measurement model as \(\y = \vec{H} \z + \e\), where \(\z \in \mathbb{R}^n\) is the unknown fully-sampled measurement and \(\e \sim \mathcal{N}(\vec{0}, \eta\vec{I})\) is Gaussian noise with variance of  $\eta$.

To estimate the posterior score, we combine the stochastic score estimate from~\eqref{eq:stochastic_denoising} with its corresponding fully-sampled prediction $\hat{\z}_{\theta}$. The posterior score is approximated as:
\[
    \begin{aligned}
    \nabla \log p_{\sigma_t}(\z_t \mid \y) 
    &= \nabla \log p_{\sigma_t}(\vec{z}_t) + \nabla \log p_{\sigma_t}(\y \mid \z_t) \\
    &\approx \nabla \log \widehat{q}_{\sigma_t}(\vec{z}_t) + \nabla \log p_{\sigma_t}(\y \mid \hat{\z}_{\theta}),
    \end{aligned}
\]
where the log-likelihood gradient is given by
\begin{equation} \label{eq:mdiff_likelihood_approximation}
    \nabla\log p_{\sigma_t}(\y \mid \hat{\z}_\theta) 
    = \gamma_t \nabla \left\| \y - \vec{H} \hat{\z}_\theta \right\|_2^2,
\end{equation}
with \(\gamma_t\) denoting the step-size parameter. Note that \(\vec{H}\) may differ from the randomized subsampling operators \(\vec{S} \sim p(\vec{S})\) used during MSM training. For example, \(\vec{H}\) can be a blurring followed by downsampling in super-resolution, or a fixed box inpainting mask. We solve the inverse problem by replacing the score function in the random walk sampling process~\eqref{eq:randomworks} with our posterior approximation. This approach avoids the need for automatic differentiation when computing the likelihood gradient, thereby reducing computational overhead. A related posterior sampling strategy, using diffusion models trained on clean data, was proposed in~\cite{wang2022ddnm}. Further simplification of our method for compressed-sensing MRI are presented in Appendix~\ref{appendix:mri_posterior_sampling}.

\subsection{Learning Partial Measurement Score from Noisy and Subsampled Measurements}

We now show how our MSM framework can be extended to train directly on noisy and subsampled measurements by integrating with self-supervised denoising methods to address the noise on the subsampled measurements.

We formulate our observed measurement as %
\[
\s = \vec{S}\z + \vec{\nu}, \quad \vec{\nu} \sim \mathcal{N}(\vec{0}, \rho\I),
\]
where $\rho$ is the measurement noise level, and the remaining notations follow Section~\ref{subsec:learning_measurement_scores}.
We first define the sequence of diffusion noise level $\{\sigma_t\}^{T}_{t=1}$.
At each time step $t$, we compare the diffusion noise level $\sigma_t$ with the measurement noise level $\rho$, and apply one of the following strategies accordingly:

\textbf{Case 1: \(\sigma_t > \rho\).}  
We add residual noise to match the diffusion level:
\[
\s_t \leftarrow \s + \sqrt{\sigma_t^2 - \rho^2} \n.
\]
The training objective is:
\[
\begin{aligned}
\mathcal{L}(\theta) &= \mathbb{E}_{\s_t, t, \vec{S}} \left[ \left\| \s - \mathbb{E}[\s \mid \s_t] \right\|_2^2 \right] + \mathcal{L}_{\text{SURE}}(\theta\,; \s, \rho) \\
&= \mathbb{E}_{\s_t, t, \vec{S}} \left[  \left\| \s - \left( \frac{\sigma_t^2 - \rho^2}{\sigma_t^2}(\hat{\s}_{\theta}(\s_t ; \sigma_t) - \s_t) + \s_t \right) \right\|_2^2 \right] + \mathcal{L}_{\text{SURE}}(\theta\,; \s, \rho),
\end{aligned}
\]
where the first term is inspired by~\cite{daras2024ambienttweedie}, which shows that a noisier image can be denoised using a less noisy reference; we apply this idea to subsampled measurements.
The second term is the SURE loss, following~\cite[Equation 9]{chen2022robustequivariant}, which enables the model to learn to denoise measurement noise and plays a key role in the next case.

\textbf{Case 2: $\sigma_t \leq \rho$.} We first denoise 
$\s$ using the MSM with the noise conditioned of $\rho$, then add diffusion noise as
\[
\s_t \leftarrow \hat{\s}_{\theta}(\s \, ; \rho) + \sigma_t\n.
\]
Training minimizes the discrepancy within the non-subsampled region:
\[
\mathcal{L}(\theta) = \mathbb{E}_{\vec{\epsilon}, \s_t, t, \vec{S}}[\hat{\s}_{\theta}(\s \, ; \rho) - \hat{\s}_{\theta}(\s_t \, ; \sigma_t) \|_2^2] + \mathcal{L}_{\text{SURE}}(\theta\,; \s,\rho).
\]
Here, $\hat{\s}_{\theta}(\s ; \rho)$ serves as a pseudo-clean reference. Its quality is crucial but improves naturally during training, since the same prediction is refined in \textbf{Case 1}. In practice, most $\sigma_t$ are larger than $\rho$, making \textbf{Case 1} more frequently sampled. As a result, the pseudo-clean reference used in \textbf{Case 2} is continuously improved, ensuring stable training across both cases.

\section{Theoretical Analysis}
\label{sec:theory}

We have introduced the MSM framework, which can generate fully-sampled measurements using \textit{partial measurement scores}.
To make the expectation of the MSM score~\eqref{eq:defined_score} efficient, our algorithm approximates it using a minibatch of \(w\) sampled operators in~\eqref{eq:ensemble}, where \(\vec{S}^{(1)}, \dots, \vec{S}^{(w)}\) are sampled independently and identically from the distribution \(p(\vec{S})\). This implies that for a fixed $\vec{W}$, we have an unbiased estimator of the MSM score:
\[
\E \left[ \nabla \log \widehat{q}_{\sigma_t}(\z_t) \right] = \nabla \log q_{\sigma_t}(\z_t),
\]
where the expectation is over the randomness in the sampled minibatch.

\begin{assumption}
\label{As:BoundedVariance}
There exists \(v > 0\) such that for all \(\z \in \mathbb{R}^n\),
\[
\E\left[\|\nabla \log q_{\sigma_t}(\z) - \nabla \log \widehat{q}_{\sigma_t}(\z)\|_2^2\right] \leq \frac{v^2}{w},
\]
where the expectation is taken over $\vec{S}^{(i)} \sim p(\vec{S})$.
\end{assumption}
This assumption implies that the gradient estimate of MSM score has a bounded variance, an assumption commonly adopted in stochastic and online algorithms~\cite{Ghadimi.Lan2016, liu2022online, welling2011bayesian}.

\begin{theorem}\label{thm:thm1}
Let \(q(\z)\) and \(\widehat{q}(\z)\) denote the distributions of samples generated by using the MSM score \(\nabla\log q_{\sigma_t}(\z_t)\) and its stochastic approximation \(\nabla\log \widehat{q}_{\sigma_t}(\z_t)\), respectively.  
Under Assumption~\ref{As:BoundedVariance}, the KL divergence between the two distributions is bounded as
\[
    \infdiv{q(\z)}{\widehat{q}(\z)} \leq \frac{v^2}{w} C,
\]
where \(C\) is a finite constant independent of \(w\).
\end{theorem}

The proof is provided in Appendix~\ref{app:proof}. Theorem~\ref{thm:thm1} quantifies the impact of using a stochastic measurement score during sampling by establishing a theoretical guarantee on the proximity of the resulting distributions. Specifically, it bounds the KL divergence between the distribution \(q(\z)\), induced by sampling with the exact MSM score, and the distribution \(\widehat{q}(\z)\), obtained using the stochastic approximation. Under Assumption~\ref{As:BoundedVariance}, the KL divergence is upper bounded by $\frac{v^2}{w} C$, where $C$ is a finite constant independent of the minibatch size. This result highlights a trade-off between computational efficiency and sampling accuracy, showing that the approximation error—measured by the KL divergence—decreases with larger minibatch size $w$.

\begin{table*}
\begin{minipage}[t]{0.49\textwidth} %
    \setstretch{1.5} %
    \centering
    \scriptsize
    \caption{\small FID scores for unconditional image samples under different training scenarios on human face images.
    \hlgreen{\textbf{Best values}} are highlighted for each training scenario, with comparisons shown when corresponding baseline methods are available. Note how MSM consistently achieves substantially lower FID scores than alternative methods across the evaluated settings.}
    \vspace{0.15cm}
    \renewcommand{\arraystretch}{1.1}
    \begin{tabular}{@{}p{2.0cm}p{3.5cm}p{0.5cm}p{0.5cm}@{}p{0.2cm}@{}p{0.5cm}@{}}
    \toprule
    \noalign{\vskip -.8ex}
    \textbf{Training data} & \multicolumn{1}{c}{\textbf{Methods}} & \multicolumn{1}{c}{\textbf{FID}$\downarrow$}  \\
    \cmidrule{1-3}  \noalign{\vskip -0.8ex}
    No degradation & \multicolumn{1}{c}{Oracle diffusion} & \multicolumn{1}{c}{$16.76$}   \\ \cdashline{1-3} 
    \multirow{2}{*}{$p = 0.4, \eta = 0$} & \multicolumn{1}{c}{MSM} & \multicolumn{1}{c}{\hlgreen{$\bm{32.54}$}}  \\
      & \multicolumn{1}{c}{Ambient diffusion} & \multicolumn{1}{c}{$61.46$}  \\ \cdashline{1-3} 
    \multirow{2}{*}{$p = 0.4, \eta = 0.1$} & \multicolumn{1}{c}{MSM} & \multicolumn{1}{c}{\hlgreen{$\bm{48.70}$}}   \\
      & \multicolumn{1}{c}{GSURE diffusion} & \multicolumn{1}{c}{$94.88$}  \\
    \bottomrule
    \end{tabular}
    \label{table:face_fid}
\end{minipage}
\hfill %
\begin{minipage}[t]{0.49\textwidth} %
    \setstretch{1.5} %
    \centering
    \scriptsize
    \caption{\small FID scores for unconditional image samples under different training scenarios on multi-coil brain MR images.
    \hlgreen{\textbf{Best values}} are highlighted for each training scenario, with comparisons shown when corresponding baseline methods are available. Note how MSM achieves a lower FID compared to the alternative methods.}
    \vspace{0.15cm}
    \renewcommand{\arraystretch}{1.5}
    \begin{tabular}{@{}p{2.0cm}p{3.5cm}p{0.5cm}p{0.5cm}@{}p{0.2cm}@{}p{0.5cm}p{0.5cm}@{}}
    \toprule
    \noalign{\vskip -1.67ex}
    \textbf{Training data} & \multicolumn{1}{c}{\textbf{Methods}} & \multicolumn{1}{c}{\textbf{FID}$\downarrow$}  \\ \noalign{\vskip -.65ex}
    \cmidrule{1-3} \noalign{\vskip -1.5ex}
    No degradation & \multicolumn{1}{c}{Oracle diffusion} & \multicolumn{1}{c}{$29.25$}  \\ \cdashline{1-3}  \noalign{\vskip -.64ex}
    \multirow{2}{*}{$R = 4, \eta = 0$} & \multicolumn{1}{c}{MSM} & \multicolumn{1}{c}{\hlgreen{$\bm{43.60}$}}  \\
      & \multicolumn{1}{c}{Ambient diffusion} & \multicolumn{1}{c}{$47.80$}  \\ \cdashline{1-3}  \noalign{\vskip -.64ex}
    $R = 4, \eta = 0.1$ & \multicolumn{1}{c}{MSM} & \multicolumn{1}{c}{$79.78$} \\  
    \bottomrule
    \end{tabular}
    \label{table:fastmri_fid}
\end{minipage}
\end{table*}

\begin{figure*}[t]
\begin{center}
\includegraphics[width=1.0\textwidth]{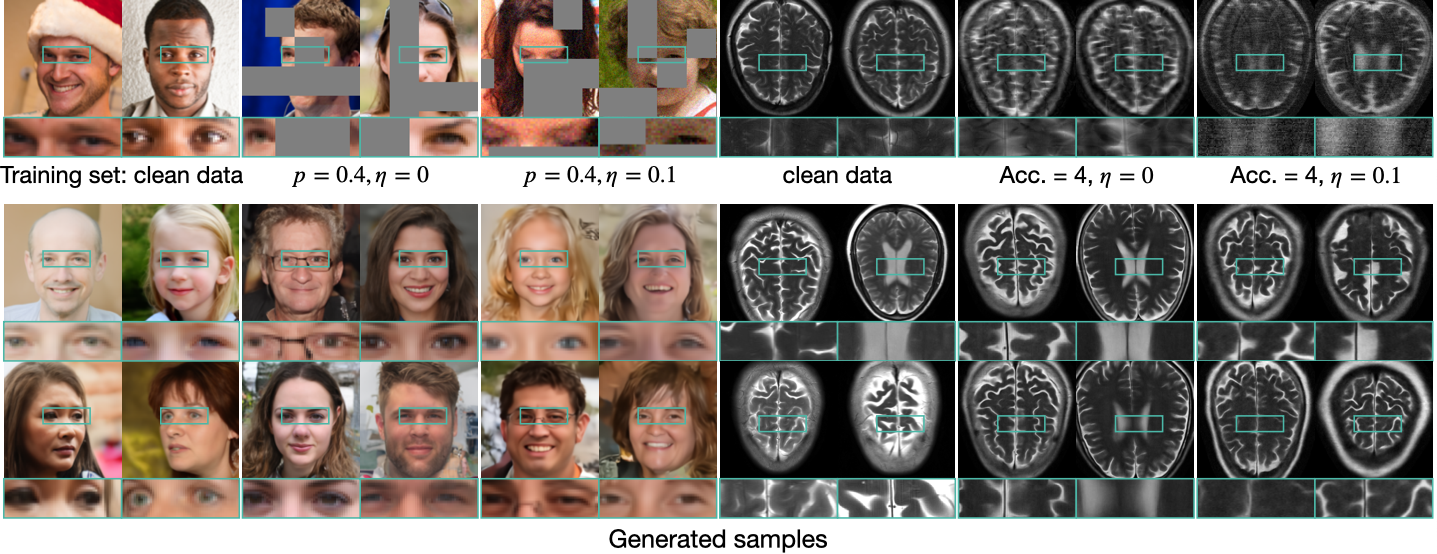}
\end{center}
\caption{Generated samples from MSM trained under three degradation settings (first row: training data; second row: samples generated by models trained on the corresponding data). Note how despite never seeing ground-truth data, MSM can generated high-quality images.}
\vspace{-.3cm}
\label{fig:visual_comparison}
\end{figure*}

\section{Numerical Evaluations}
\label{sec:numericalevaluations}

We evaluated MSM on unconditional generation and conditional sampling for natural images and multi-coil MRI.  All models used the same diffusion architecture~\cite{dhariwal2021beat} and were trained from scratch on a single NVIDIA A100 GPU for 1,000,000 iterations (see Appendix~\ref{appendix:training_details} for architectural details). 
Natural image experiments used 69,000 FFHQ face images ($128\times128$ RGB), and MRI experiments used 2,000 center-cropped T2-weighted slices ($256\times256$, complex-valued) from the fastMRI dataset~\cite{zbontar2018fastmri1, knoll2020fastmri2}, which was curated under Institutional Review Board (IRB) approval and deidentified prior to release.
We performed inverse problem evaluations on 100 test images per domain.

\subsection{RGB Face Image Experiment}
\label{subsec:rgb_experiment}

\textbf{Training data.}  
We considered two training settings: (1) subsampling only and (2) subsampling with added Gaussian noise \(\eta = 0.1\).  
In both cases, 40\% of pixels were masked using \(32 \times 32\) patches, applied identically across RGB channels (see Figure~\ref{fig:visual_comparison} for an example).  
Full training details of our method and baselines are provided in Appendix~\ref{appendix:training_details} and Appendix~\ref{appendix:baselines_for_uncond}.

\textbf{Unconditional sampling.}  
MSM used a stochastic loop parameter \(w = 3\) with 200 sampling steps.
We compared MSM to three baselines: an oracle diffusion model trained on clean images, Ambient diffusion~\cite{daras2024ambientdiffusion} trained on noiseless masked inputs, and GSURE diffusion~\cite{kawar2023gsure} trained on noisy masked inputs.
All baselines used 200-step accelerated sampling~\cite{song2021ddim}.
As shown in Table~\ref{table:face_fid}, MSM achieves better FID scores than all baselines trained without clean data, evaluated over 3,000 generated samples.
Figure~\ref{fig:visual_comparison} further shows that MSM generates clean images despite being trained without clean data.
Additional results showing how \(w\) influences sampling quality and time efficiency are provided in Appendix~\ref{appendix:effect_of_stochastic_loop}.

\begin{table*}
\begin{minipage}[t]{0.45\textwidth} %
    \setstretch{1.5} %
    \centering
    \scriptsize
    \caption{\small Quantitative results on two natural image inverse problems comparing methods using diffusion priors trained without clean images. \hlgreen{\textbf{Best values}} are highlighted per metric. MSM achieves the best performance across both distortion-based and perception-oriented metrics.}
    \vspace{0.15cm}
    \renewcommand{\arraystretch}{.9}
    \begin{tabular}{@{}p{0.9cm}p{.1cm}p{.1cm}p{.1cm}p{.1cm}@{}}
    \toprule
    \noalign{\vskip -.8ex}
\textbf{Setup} &  & \multicolumn{1}{c}{\text{Input}} & \multicolumn{1}{c}{\text{A-DPS}} & \multicolumn{1}{c}{\text{MSM}} \\
    \cmidrule{1-5}\\ \noalign{\vskip -3.5ex}
    \multirow{3}{*}{{\textbf{Inpainting}}} & \multicolumn{1}{c}{PSNR$\uparrow$} & \multicolumn{1}{c}{$18.26$} & \multicolumn{1}{c}{$20.14$} & \multicolumn{1}{c}{{\hlgreen{$\bm{24.71}$}}}  \\[+.95ex]
     & \multicolumn{1}{c}{SSIM$\uparrow$} & \multicolumn{1}{c}{$0.749$} & \multicolumn{1}{c}{$0.621$} & \multicolumn{1}{c}{{\hlgreen{$\bm{0.867}$}}}   \\[+.75ex]
      & \multicolumn{1}{c}{LPIPS$\downarrow$} & \multicolumn{1}{c}{$0.304$} & \multicolumn{1}{c}{$0.305$} & \multicolumn{1}{c}{{\hlgreen{$\bm{0.076}$}}}  \\[+.5ex] 
      \cdashline{1-5} \\ \noalign{\vskip -2.5ex} 
      \multirow{3}{*}{{\textbf{SR} ($\times 4$)}} & \multicolumn{1}{c}{PSNR$\uparrow$} & \multicolumn{1}{c}{$23.21$} & \multicolumn{1}{c}{$22.61$} & \multicolumn{1}{c}{{\hlgreen{$\bm{28.11}$}}}   \\[+.95ex]
     & \multicolumn{1}{c}{SSIM$\uparrow$} & \multicolumn{1}{c}{$0.728$} & \multicolumn{1}{c}{$0.702$} & \multicolumn{1}{c}{{\hlgreen{$\bm{0.868}$}}}   \\[+.75ex]
      & \multicolumn{1}{c}{LPIPS$\downarrow$} & \multicolumn{1}{c}{$0.459$} & \multicolumn{1}{c}{$0.277$} & \multicolumn{1}{c}{{\hlgreen{$\bm{0.117}$}}} \\[+.5ex] 
      \bottomrule
    \end{tabular}
    \label{table:quantitative_comparison_FFHQinvproblems}
\end{minipage}
\hfill %
\begin{minipage}[t]{0.52\textwidth} %
    \setstretch{1.5} %
    \centering
    \scriptsize
    \caption{\small Quantitative results on multi-coil compressed sensing MRI comparing diffusion-based and self-supervised methods, all trained without clean data. \hlgreen{\textbf{Best values}} are highlighted per metric. MSM outperforms the baselines in both PSNR and LPIPS, including the restoration-specific baseline SSDU.}
    \vspace{0.15cm}
    \renewcommand{\arraystretch}{.9}
    \begin{tabular}{@{}p{1.3cm}p{.1cm}p{.1cm}p{.1cm}p{.1cm}p{.1cm}@{}}
    \toprule
    \noalign{\vskip -.8ex}
\textbf{Setup} &  & \multicolumn{1}{c}{\text{Input}} & \multicolumn{1}{c}{\text{A-DPS}} & \multicolumn{1}{c}{\text{SSDU}} & \multicolumn{1}{c}{\text{MSM}} \\
    \cmidrule{1-6}\\ \noalign{\vskip -3.5ex}
        \multirow{3}{*}{{\textbf{CS-MRI} ($\times 4$)}} & \multicolumn{1}{c}{PSNR$\uparrow$} & \multicolumn{1}{c}{$22.75$} & \multicolumn{1}{c}{$27.28$}  & \multicolumn{1}{c}{$29.65$} & \multicolumn{1}{c}{{\hlgreen{$\bm{30.71}$}}}   \\[+.95ex]
     & \multicolumn{1}{c}{SSIM$\uparrow$} & \multicolumn{1}{c}{$0.648$} & \multicolumn{1}{c}{$0.804$} & \multicolumn{1}{c}{{\hlgreen{$\bm{0.847}$}}} & \multicolumn{1}{c}{$0.839$}   \\[+.75ex]
      & \multicolumn{1}{c}{LPIPS$\downarrow$} & \multicolumn{1}{c}{$0.306$} & \multicolumn{1}{c}{$0.173$} & \multicolumn{1}{c}{$0.160$} & \multicolumn{1}{c}{{\hlgreen{$\bm{0.145}$}}} \\[+.5ex] 
      \cdashline{1-6} \\ \noalign{\vskip -2.5ex}
    \multirow{3}{*}{{\textbf{CS-MRI} ($\times 6$)}} & \multicolumn{1}{c}{PSNR$\uparrow$} & \multicolumn{1}{c}{$21.94$} & \multicolumn{1}{c}{$26.29$} & \multicolumn{1}{c}{$28.02$} & \multicolumn{1}{c}{{\hlgreen{$\bm{28.86}$}}}  \\[+.95ex]
     & \multicolumn{1}{c}{SSIM$\uparrow$} & \multicolumn{1}{c}{$0.617$} & \multicolumn{1}{c}{$0.763$} & \multicolumn{1}{c}{{\hlgreen{$\bm{0.820}$}}} & \multicolumn{1}{c}{$0.805$}   \\[+.75ex]
      & \multicolumn{1}{c}{LPIPS$\downarrow$} & \multicolumn{1}{c}{$0.342$} & \multicolumn{1}{c}{$0.201$} & \multicolumn{1}{c}{$0.186$} & \multicolumn{1}{c}{{\hlgreen{$\bm{0.168}$}}}  \\[+.5ex]  
      \bottomrule
    \end{tabular}
    \label{table:quantitative_comparison_MRIinvproblems}
\end{minipage}
\end{table*}

\begin{figure*}[t]
\vspace{-.cm}
\begin{center}
\includegraphics[width=1.0\textwidth]{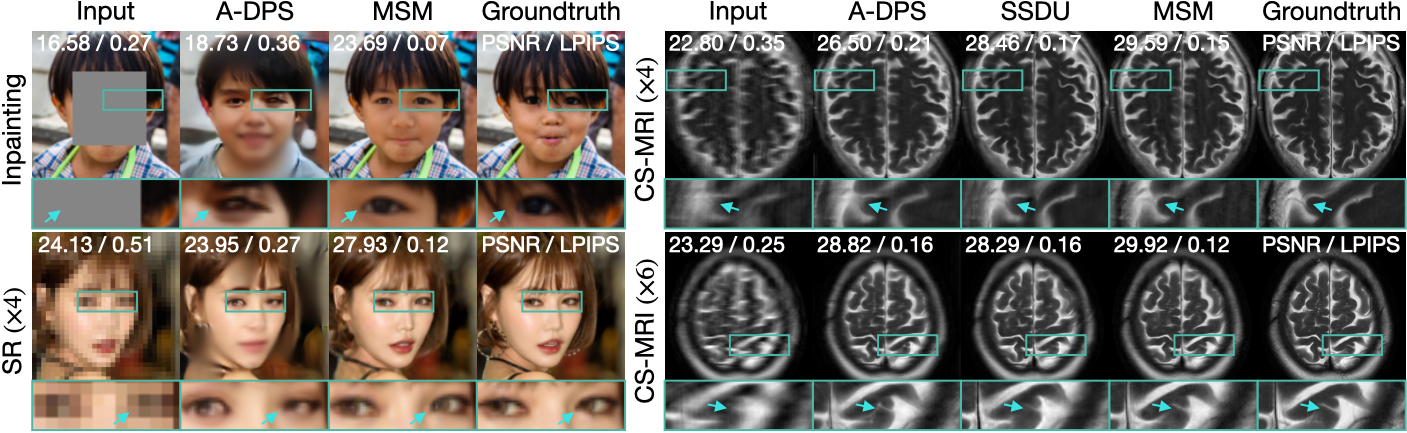}
\end{center}
\caption{Visual comparison of methods trained on subsampled data for inverse problems. Note how MSM leads to the best results in both applications.}
\vspace{-.3cm}
\label{fig:visual_comparison_invproblems}
\end{figure*}

\textbf{Conditional sampling.}  
We compared MSM with A-DPS~\cite{aali2025ambientposterior} on two inverse problems: box inpainting with a \(64 \times 64\) missing region and \(4\times\) bicubic super-resolution.  
Both methods used models trained on the same noiseless 40\% masked data without retraining.  
A-DPS took 1000 steps, while MSM took 200 steps with \(w = 3\), and step sizes in \eqref{eq:mdiff_likelihood_approximation} set to \(\gamma_t = 1.75\) for both inpainting and super-resolution.  
As shown in Table~\ref{table:quantitative_comparison_FFHQinvproblems} and Figure~\ref{fig:visual_comparison_invproblems}, MSM outperforms A-DPS in PSNR, SSIM, and LPIPS.  
We also observe that A-DPS performs worse than the input image in PSNR and SSIM, likely due to limitations of the Ambient diffusion prior in generating fine details when trained on non-sparse subsampling patterns—%
such as box masks—%
rather than more localized patterns (e.g., dust masks) mainly used in the original training setup of Ambient diffusion~\cite{daras2024ambientdiffusion}.
Detailed hyperparameter setups for A-DPS are provided in Appendix~\ref{appendix:baselines_for_cond}.
Additional experiments showing that MSM achieves comparable performance to clean-data-trained diffusion-based inverse problem solvers are provided in Appendix~\ref{app:comparison_inv_wth_DIS}.
Results using an MSM prior trained on noisy and subsampled data are provided in Appendix~\ref{appendix:solve_inv_wth_noisy_mdiff}.

\subsection{Multi-Coil MRI Experiment}
\label{subsec:mri_experiment}

\textbf{Training data.}  
We considered two training settings: (1) subsampling only and (2) subsampling with added Gaussian noise \(\eta = 0.1\).  
We applied a subsampling operation in k-space using random masks with an acceleration rate of \(R = 4\), while fully sampling all vertical lines and the central 20 lines for autocalibration.  
Training configuration details for our method and baselines are provided in Appendix~\ref{appendix:training_details} and Appendix~\ref{appendix:baselines_for_uncond}.

\textbf{Unconditional sampling.} MSM used a stochastic loop parameter of \(w = 2\) with 200 sampling steps (see Appendix~\ref{appendix:effect_of_stochastic_loop} for how \(w\) affects sampling quality).
We compared it against an oracle diffusion model trained on clean images and Ambient diffusion~\cite{daras2024ambientdiffusion} trained on noiseless subsampled inputs; GSURE diffusion~\cite{kawar2023gsure} is omitted due to incompatibility with multi-coil data.
All baselines used accelerated sampling with 200 steps~\cite{song2021ddim}.
As shown in Table~\ref{table:fastmri_fid}, MSM achieves better FID scores than Ambient diffusion, based on 3,000 generated samples.
Figure~\ref{fig:visual_comparison} further demonstrates that MSM generates realistic images, even when trained on noisy and subsampled measurements.
FID computation details are provided in Appendix~\ref{app:fid}, and generation results under extreme subsampling are shown in Appendix~\ref{appendix:extreme}.

\textbf{Conditional sampling.}  
We evaluated accelerated MRI reconstruction using random masks with acceleration rates \(R = 4\) and \(R = 6\), and measurement noise \(\eta = 0.01\).  
We used stochastic posterior sampling algorithm introduced in Appendix~\ref{appendix:mri_posterior_sampling} with MSM pretrained on noiseless \(R = 4\) measurements without retraining, with step size \(\gamma_t = 2\) in \eqref{eq:mdiff_subsampling_likelihood_approximation}.  
We compared against two baselines trained on the same subsampled data: a diffusion-based method (A-DPS~\cite{aali2025ambientposterior}) and a self-supervised end-to-end method (Robust SSDU~\cite{millard2024robustssdu}).  
A-DPS used 1000 steps, SSDU performed a single forward pass, and MSM used 200 steps with \(w = 3\).  
As shown in Table~\ref{table:quantitative_comparison_MRIinvproblems} and Figure~\ref{fig:visual_comparison_invproblems}, MSM outperforms both baselines in PSNR and LPIPS.
Detailed hyperparameter setups for A-DPS and Robust SSDU are provided in Appendix~\ref{appendix:baselines_for_cond}.
Comparisons with clean-data-trained diffusion-based inverse problem solvers are in Appendix~\ref{app:comparison_inv_wth_DIS}, and additional results using an MSM prior trained on noisy and subsampled MRI data are in Appendix~\ref{appendix:solve_inv_wth_noisy_mdiff}.

\section{Conclusion}
\label{sec:Conclusion}

We introduced \emph{Measurement Score-based diffusion Model (MSM)}, a framework that can model the full measurement distribution using score functions learned solely from noisy, subsampled measurements.
The key idea is the MSM score, defined as an expectation over partial scores induced by randomized subsampling.
We develop a stochastic sampling algorithm for both prior and posterior inference that efficiently approximates this expectation, enabling clean image generation and inverse problem solving.
We demonstrate that MSM achieves state-of-the-art performance among diffusion-based methods trained without clean data, for both unconditional image generation and conditional sampling in inverse problems.
The framework applies broadly to settings where only subsampled measurements are available but collectively cover the full data space, making it valuable for generative modeling in limited-data regimes and high-dimensional sampling from low-dimensional observations.

{
\small

\bibliographystyle{IEEEbib}
\bibliography{refs}
}

\newpage
\appendix

\section{Proofs}\label{app:proof}
\begin{suptheorem}
Let \(q(\z)\) and \(\widehat{q}(\z)\) denote the distributions of samples generated by using the MSM score \(\nabla\log q_{\sigma_t}(\z)\) and its stochastic approximation \(\nabla\log \widehat{q}_{\sigma_t}(\z)\), respectively.  
Under Assumption~\ref{As:BoundedVariance}, the KL divergence between the two distributions is bounded as
\begin{equation}
   \infdiv{q(\z)}{\widehat{q}(\z)} \leq \frac{v^2}{w} C,
\end{equation}
where \(C\) is a finite constant independent of \(w\).
\end{suptheorem}
\begin{proof}
Our proof invokes Girsanov’s theorem, which characterizes how the distribution of a Brownian‑driven stochastic process transforms when we transition from one probability measure to another~\cite{chen2024accelerating,baker2024condition,huang2021a,song2021maximum}.

Consider the two  stochastic processes $\{\z (t)\}_{t\in[0,1]}$ and $\{\widehat{\z} (t)\}_{t\in[0,1]}$, corresponding to the Euler–Maruyama discretizations of the following reverse-time SDEs
\begin{align*}
        \d \z &= \nabla \log q_{\sigma_t}(\z) \d t + \sqrt{2\mathcal{T}_t} \d \overline{\vec{w}}_t\quad \z(T) = \z_T \sim \mathcal{N}(\vec{0}, \vec{I}),   \\
 \d \widehat{\z} &= \nabla \log \qhat_{\sigma_t}(\zhat) \d t  + \sqrt{2\mathcal{T}_t} \d \overline{\vec{w}}_t\quad \widehat{\z}(T) =  \widehat{\z}_T \sim \mathcal{N}(\vec{0}, \vec{I}),   \\
\end{align*}
Let $\Q [\z_{0:T}]$ and $\Qhat[\zhat_{0:T}]$ denote the path measures induced by the respective processes. The process $\{\z (t)\}_{t\in[0,1]}$  corresponds to the reverse diffusion trajectory driven by the true measurement score $\nabla \log q_{\sigma_t}(\z)$, while   $\{\widehat{\z} (t)\}_{t\in[0,1]}$ is generated by the reverse process using an approximate (stochastic) measurement score $\nabla \log \qhat_{\sigma_t}(\z)$. 

 By using the chain rule of KL divergence from Lemma~\ref{Lem:ChainRule}, we have 
\begin{equation}\label{eq:KLpathmeasures}
    \infdiv{\Q}{\Qhat} = \E_{\z_T\sim \mathcal{N}(\vec{0}, \vec{I}) }  \left [\infdiv{\Q(.|\z= \z_T)}{\Qhat(.|\zhat= \zhat_T)} \right ].
\end{equation}

Using the definition of KL divergence and the fact that $M_T = \d \Qhat/\d\Q$ from Lemma~\ref{lem:GirsavonTheorem}, we have 
\begin{align*}
&\infdiv{\Q(.|\z= \z_T)}{\Qhat(.|\zhat= \zhat_T)} =  -\E_\Q \left [ \log \frac{\d \Qhat}{\d \Q}\right] = -\E_\Q \left [ \log M_T\right]  \\
&~= \E_\Q\Bigg [ \int_0^T   \frac{\big [ \nabla \log q_{\sigma_t}(\z_t) - \nabla \log \qhat_{\sigma_t}(\z_t) \big] }{\sqrt{2\T_t} } \d \overline{\vec{w}}_t \quad + \frac{1}{2} \int_0^T  \frac{ \| \nabla \log q_{\sigma_t}(\z_t) - \nabla \log \qhat_{\sigma_t}(\z_t) \|_2^2 }{2\T_t } \d t
\Bigg ]\\ 
&~= \E_\Q\Bigg [ \int_0^T   \frac{\big [ \nabla \log q_{\sigma_t}(\z_t) - \nabla \log \qhat_{\sigma_t}(\z_t) \big] }{\sqrt{2\T_t} } \d \overline{\vec{w}}_t \Bigg ] + \E_\Q\Bigg [ \frac{1}{2} \int_0^T  \frac{ \| \nabla \log q_{\sigma_t}(\z_t) - \nabla \log \qhat_{\sigma_t}(\z_t) \|_2^2 }{2\T_t } \d t
\Bigg ]\\
&~= \E_\Q\Bigg [\int_0^T   \frac{1}{\sqrt{2\T_t}}~\E\Big[\nabla \log q_{\sigma_t}(\z_t)  - \frac{1}{w}\sum_{i=1}^{w}  \nabla  \log q_{\sigma_t}(\s_t^{(i)})\Big|_{\vec{s}_t^{(i)}= \vec{S}^{(i)} \vec{z}_t} \Big]\d \overline{\vec{w}}_t\Bigg ]\\
& ~+\E_\Q\Bigg [\int_0^T   \frac{1}{4\T_t} ~\E\Big[\Big\|\nabla \log q_{\sigma_t}(\z_t)   - \frac{1}{w}\sum_{i=1}^{w}   \nabla \log q_{\sigma_t}(\s_t^{(i)})\Big|_{\vec{s}_t^{(i)}= \vec{S}^{(i)} \vec{z}_t}\Big\|_2^2\Big]\d t\Bigg ]\\
& ~=\E_\Q\Bigg [\int_0^T   \frac{1}{4\T_t} ~\E\Big[\Big\|\nabla \log q_{\sigma_t}(\z_t) - \frac{1}{w}\sum_{i=1}^{w}   \nabla \log q_{\sigma_t}(\s_t^{(i)})\Big|_{\vec{s}_t^{(i)}= \vec{S}^{(i)} \vec{z}_t}\Big\|_2^2\Big]\d t\Bigg ] \leq \frac{v^2}{w} \int_0^T \frac{1}{4\T_t} \d t \leq \frac{v^2}{w} C,
\end{align*}
where $C \defn \int_0^T   1/(4\T_t) \d t $ is a finite constant. 
In the first line, we use the definition of KL divergence between $\Q$ and $\Qhat$ and the result from Lemma~\ref{lem:GirsavonTheorem}. In the third line, we use the law of iterated expectations over \( w \) sampling masks \( \vec{S}^{(i)} \sim p(\vec{S}) \) for \( i = 1, \dots, w \). Note that since $\E [\nabla \log \qhat_{\sigma_t}
(\z_t)]$ is an unbiased estimator of MSM score $\nabla \log q_{\sigma_t}(\z_t)$, we have $ \E [\nabla \log \qhat_{\sigma_t}
(\z_t)] = \nabla \log q_{\sigma_t}(\z_t) $, which yields the expectation in the forth line to be $0$. 
In the last line, we use the bounded variance in Assumption~\ref{As:BoundedVariance}.

Following this result with \eqref{eq:KLpathmeasures} and Lemma~\ref{Lem:inequalitymarginalized}, we have 
\begin{equation}
   \infdiv{q_0}{\qhat_0} \leq\infdiv{\Q}{\Qhat} \leq \frac{v^2}{w} C.
\end{equation}
\end{proof}

\begin{lemma}\label{lem:GirsavonTheorem}
\textbf{(The Girsanov Theorem III.)}
Let $\{\z(t)\}_{t = T}^0$ and $\{\zhat(t)\}_{t = T}^0$ be two It\^{o} process of the forms
\begin{align*}
        \d \z &= \nabla \log q_{\sigma_t}(\z) \d t + \sqrt{2\mathcal{T}_t} \d \overline{\vec{w}}_t\quad \z(T) = \z_T \sim \mathcal{N}(\vec{0}, \vec{I})   \\
 \d \widehat{\z} &= \nabla \log \qhat_{\sigma_t}(\zhat) \d t  + \sqrt{2\mathcal{T}_t} \d \overline{\vec{w}}_t\quad \widehat{\z}(T) =  \widehat{\z}_T \sim \mathcal{N}(\vec{0}, \vec{I}),   \\
\end{align*}

where $0\leq T\leq \infty$ is a given constant,  and $\overline{\vec{w}} \in \R^n$ is a $n-$dimensional Brownian motion. 
Suppose that there exist a process $\alpha (\z, t)$ such that
\begin{equation}\label{eq:alphadef}
    \alpha(\z, t) = \frac{\big [ \nabla \log q_{\sigma_t}(\z_t) - \nabla \log \qhat_{\sigma_t}(\z_t) \big] }{\sqrt{2\T_t} },
\end{equation}
which satisfies Novikov's condition 
\begin{equation*}
    \E\left [ exp\left (\frac{1}{2}\int_0^T\alpha^2(\z, t) \d t\right ) \right ] = \E\left [ exp\left (\frac{1}{2}\int_0^T \frac{ \| \nabla \log q_{\sigma_t}(\z_t) - \nabla \log \qhat_{\sigma_t}(\z_t) \|_2^2 }{2\T_t } \d t \right ) \right ] < \infty, 
\end{equation*}
where $\E = \E_\Q$ is the expectation with respect to $\Q$, probability measure induced by the process $\{\z_t\}_{t = T}^0$. Then, we can define $M_T$ and probability measure $\Qhat$, induced by process $\{\zhat_t\}_{t = T}^0$ as 
\begin{align*}\label{eq:Radon_Nikodym}
   & M_T\defn \\& \exp \Bigg ( - \int_0^T   \frac{\big [ \nabla \log q_{\sigma_t}(\z_t) - \nabla \log \qhat_{\sigma_t}(\z_t) \big] }{\sqrt{2\T_t} } \d\overline{\vec{w}}_t - \frac{1}{2}\int_0^T  \frac{ \| \nabla \log q_{\sigma_t}(\z_t) - \nabla \log \qhat_{\sigma_t}(\z_t) \|_2^2 }{2\T_t } \d t \Bigg ),
\end{align*}
where 
\begin{equation}
   t \leq T \quad \text{and } \quad \d \Qhat \defn  M_T \d \Q. 
\end{equation}
\end{lemma}
Proof of the Girsanov Theorems can be found in \cite[Theorems 8.6.3, 8.6.4, and 8.6.5]{oksendal2013stochastic}.

\textbf{Remark.} Note that it can be shown that Novikov's condition is satisfied for function $\alpha(\z, t)$ in \eqref{eq:alphadef} as
\begin{align*}
&\E\left [ exp\left (\frac{1}{2}\int_0^T \frac{ \| \nabla \log q_{\sigma_t}(\z_t) - \nabla \log \qhat_{\sigma_t}(\z_t) \|_2^2 }{2\T_t } \d t \right ) \right ] \\
& = \E\left [ exp\left (\frac{1}{2}\int_0^T \frac{ \E\left[\| \nabla \log q_{\sigma_t}(\z_t) - \nabla \log \qhat_{\sigma_t}(\z_t) \|_2^2 \right] }{2\T_t } \d t\right ) \right ] \\
    &\quad \leq \frac{v^2}{w} \cdot exp\left (\frac{1}{2}\int_0^T \frac{1}{2\T_t } \d t \right )   < \infty, 
\end{align*}
where in the second line, we use the total law of expectation  (i.e., $\E[\vec{a}] = \E[ \E[\vec{a}|\vec{b}]]$) we use the law of iterated expectations over \( w \) sampling masks \( \vec{S}^{(i)} \sim p(\vec{S}) \) for \( i = 1, \dots, w \). Here, we use Assumption~\ref{As:BoundedVariance}and the fact the $\int_0^T (1/(2\T_t)) ~\d t$ is a finite constant.

\begin{lemma}
    \label{Lem:inequalitymarginalized}
    Let $\Q$ and $\Qhat$ be the path measure of two stochastic processes $\{\z(t)\}_{t = 0}^T$ and $\{\zhat(t)\}_{t = 0}^T$. We denote $q_0$ and $\qhat_0$ as the marginal distribution of $\z(0)$ and $\zhat(0)$. Then, we have 
    \begin{equation*}
        \infdiv{q_0}{\qhat_0} \leq\infdiv{\Q}{\Qhat}. 
    \end{equation*}
\end{lemma}
\begin{proof}
From the chain rule of KL divergence, we have 
\begin{align*}
    \infdiv{\Q}{\Qhat}  &= \infdiv{\Q_{\z(0)= \z_0}}{\Qhat_{\zhat(0) =\z_0 }}\\
    &+ \int_{\z}  \infdiv{\Q(.\mid \z(0)= \z_0)}{\Qhat(.\mid \zhat(0) =\z_0 )} \Q_{\z(0)= \z_0}(\d \z) \\
    & =  \infdiv{q_0}{\qhat_0} +\int_{\z}  \infdiv{\Q(.\mid \z(0)= \z_0)}{\Qhat(.\mid \zhat(0) =\z_0 )} \Q_{\z(0)= \z_0}(\d \z).
\end{align*}
From the non-negativity of KL divergence, we obtain the desired results. 
\end{proof}

\begin{lemma}\label{Lem:ChainRule}
\textbf{(Chain Rule of KL Divergence.)}

Let $\Q$ and $\Qhat$ be the path measure induced by the two following reverse-time SDEs
\begin{align*}
       \d \z &= \nabla \log q_{\sigma_t}(\z) \d t  + \sqrt{2\mathcal{T}_t} \d \vec{w}_t\quad \z(T) = \z_T \sim \mathcal{N}(\vec{0}, \vec{I})  \\
 \d \widehat{\z} &= \nabla \log \qhat_{\sigma_t}(\zhat) \d t  + \sqrt{2\mathcal{T}_t} \d \vec{w}_t\quad \widehat{\z}(T) =  \widehat{\z}_T \sim \mathcal{N}(\vec{0}, \vec{I}).   \\
\end{align*}
From the chain rule of KL divergence, we have 
\begin{align*}
\infdiv{\Q}{\Qhat} &= \infdiv{\Q_{\z(T)= \z_T}}{\Qhat_{\zhat(T)= \zhat_T}} \\
   & + \int_{\z}  \infdiv{\Q(.|\z(T) =\z_T)}{\Qhat(.| \zhat(T) =\zhat_T)}\,\Q_{\z(T) = \z_T} (\d \z)\\
   &= \infdiv{\Q_{\z(T)= \z_T}}{\Qhat_{\zhat(T)= \zhat_T}}  + \E_{\z_T\sim \mathcal{N}(\vec{0}, \vec{I}) }  \left [\infdiv{\Q(.|\z= \z_T)}{\Qhat(.|\zhat= \z_T)} \right ] \\
   & = \E_{\z_T\sim \mathcal{N}(\vec{0}, \vec{I}) }  \left [\infdiv{\Q(.|\z= \z_T)}{\Qhat(.|\zhat= \z_T)} \right ],
\end{align*}
where in the last two equalities, we use the fact that $\Q_{\z_T} = \Qhat_{\z_T} = \mathcal{N}(\vec{0}, \vec{I})$. 
\end{lemma}
\newpage

\section{Implementation Details}

\subsection{MSM Framework in MR Images}
\label{subsec:mri_framework}

The main manuscript illustrates the measurement score-based diffusion model \textit{(MSM)} framework’s training and sampling schemes, but omits domain-specific transformations for clarity. These transformations are essential in the MRI setting, which requires conversions between measurement and image spaces before and after denoising.

Specifically, we apply the inverse Fourier transform \(\vec{F}^{\top}\) followed by the adjoint coil-sensitivity operator \(\vec{C}^{\top}\) to project the measurements into image space before denoising. After denoising, we map the denoised image back to measurement space by applying the forward coil-sensitivity operator \(\vec{C}\) and the Fourier transform \(\vec{F}\).

This results in a modified version of Equation~\ref{eq:denoised_sub_measurement}, given by:
\begin{equation}\label{eq:denoised_sub_measurement_mri}
\hat{\s}_{\theta}(\s_t\,;\, \sigma_t) \leftarrow \vec{S} \vec{FC}\,\mathsf{D}_{\theta}(\vec{C}^{\top}\vec{F}^{\top}\s_t \,;\, \sigma_t),
\end{equation}
where \(\vec{S}\) is the subsampling operator. A visual illustration is provided in Figure~\ref{fig:illustration_mri}.

\begin{figure*}[htbp]
\begin{center}
\includegraphics[width=1.0\textwidth]{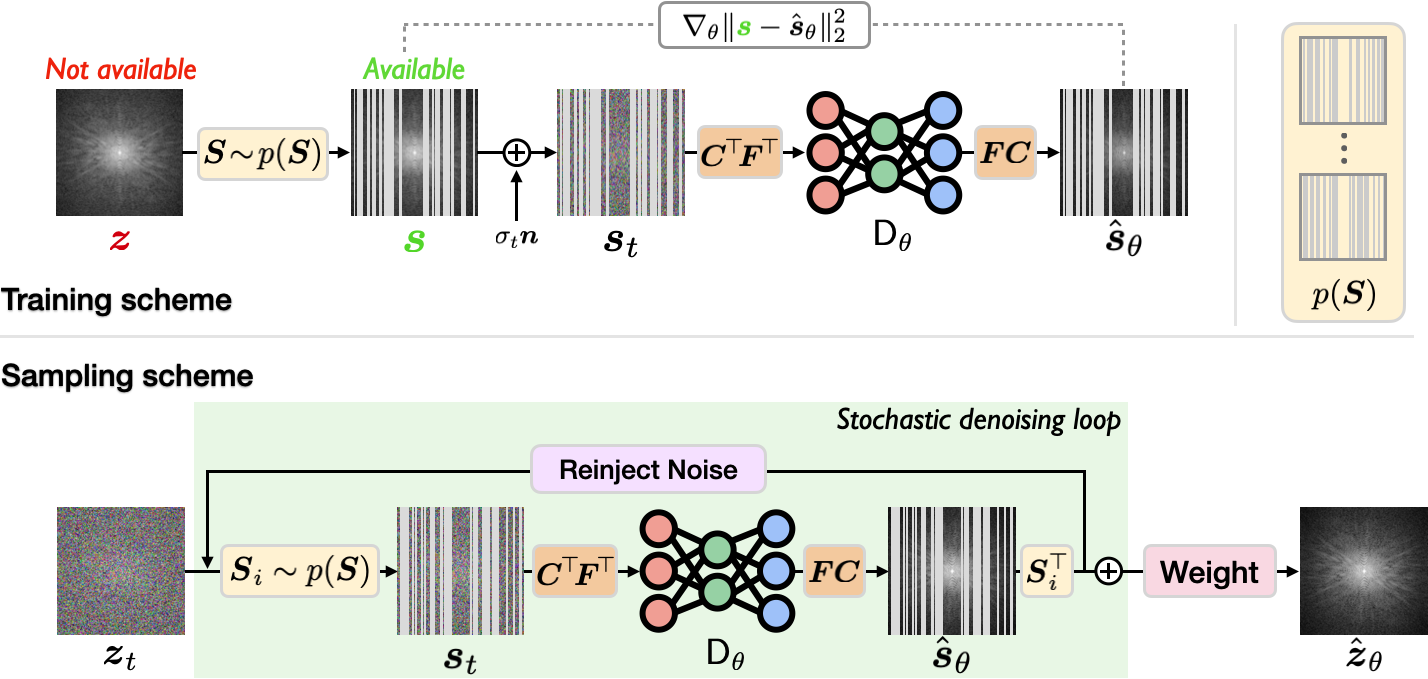}
\end{center}
\caption{Illustration of the \emph{Measurement Score-based diffusion Model (MSM)} for training and sampling using subsampled MRI measurements. MSM operates directly on k-space measurements, with minor domain transformations between k-space and image space only at the input and output of the diffusion model $\mathsf{D}_{\theta}$.}
\vspace{-.3cm}
\label{fig:illustration_mri}
\end{figure*}

\subsection{Stochastic Posterior Sampling for Compressed-Sensing MRI}
\label{appendix:mri_posterior_sampling}

Our posterior sampling algorithm is described in Section~\ref{subsec:measurementdiff_posterior_sampling}.
For compressed-sensing MRI, we apply an additional simplification based on directly approximating the \emph{partial posterior score} for each partially subsampled measurement within the stochastic algorithm.

To perform posterior sampling for compressed-sensing MRI, we estimate the posterior score using a stochastic ensemble, similar to the MSM score ensemble in Section \ref{subsec:measurementdiff_prior_sampling}:
\begin{equation}
\begin{aligned}
    \nabla \log p_{\sigma_t}(\z_t \mid \y) 
    &= \nabla \log p_{\sigma_t}(\vec{z}_t) + \nabla \log p_{\sigma_t}(\y \mid \z_t) \\
    &\approx  \vec{W} \left[\frac{1}{w} \sum_{i=1}^{w} \vec{S}^{(i)\top} \left( \nabla \log p_{\sigma_t}(\s_t^{(i)}) + \nabla \log p_{\sigma_t}(\y^{(i)} \mid \s_t^{(i)})\right)\Big|_{\vec{s}_t^{(i)} = \vec{S}^{(i)} \vec{z}_t}\right],
\end{aligned}
\end{equation}
where for each subsampling operator \(\vec{S}^{(i)} \in \mathbb{R}^{m_i \times n}\), we define \(\y^{(i)} = \vec{S}^{(i)} \vec{H}^{\top} \y\) to project the observed measurement \(\y\) into the same partial measurement space as \(\s_t^{(i)}\).
The log-likelihood gradient is approximated by
\begin{equation} \label{eq:mdiff_subsampling_likelihood_approximation}
\begin{aligned}
    \nabla\log p_{\sigma_t}(\y^{(i)} \mid \s^{(i)}_t) 
    &\approx \nabla \log p_{\sigma_t}(\y^{(i)} \mid \hat{\s}_{\theta}^{(i)}) \\
    &= \gamma_t \nabla \left\| \y^{(i)} - \tilde{\vec{H}}^{(i)}\hat{\s}_\theta^{(i)} \right\|_2^2,
\end{aligned}
\end{equation}
where \( \gamma_t \) is a tunable step size, and we define $\tilde{\vec{H}}^{(i)} = \vec{S}^{(i)} \vec{H}^{\top} \vec{H} \vec{S}^{(i)\top}$ as the degradation operator $\vec{H}^{\top} \vec{H} $ restricted to the coordinates selected by $\vec{S}^{(i)}$.

\subsection{Model Architecture and Training Configuration}
\label{appendix:training_details}

We adopted the U-Net architecture~\cite{ronneberger2015unet}, following the design used in~\cite{ho_NEURIPS2020_ddpm, dhariwal2021beat}, as our diffusion model backbone.
Models were trained with the AdamW optimizer~\cite{loshchilov2019adamw} and used an exponential moving average (EMA) to stabilize training by averaging model weights over time, using a decay rate of 0.9999 for gradual updates.
The diffusion process consisted of 1000 timesteps, with a linearly increasing noise variance schedule starting from 0.0001 and reaching 0.2 at the final step.
All diffusion models, including diffusion-based baselines, were trained with the same architecture for each application. The hyperparameter setup and architectural details are summarized in Table~\ref{table:ablation_model_details}.
\begin{table}[htbp]
    \setstretch{1.5} %
    \centering
    \scriptsize
    \caption{\small Diffusion model architecture and training hyperparameters for each dataset.}
    \vspace{0.15cm}
    \renewcommand{\arraystretch}{1.0}
    \begin{tabular}{@{}p{2.8cm}p{2.5cm}p{0.5cm}p{0.5cm}@{}p{0.2cm}@{}p{0.5cm}p{0.5cm}@{}}
    \toprule
    \noalign{\vskip -1.67ex}
     & \multicolumn{1}{c}{\textbf{RGB Face Images}} & \multicolumn{1}{c}{\textbf{Multi-Coil MRI}}  \\ \noalign{\vskip -.65ex}
    \cmidrule{1-3} \noalign{\vskip -1.5ex}
    \textbf{Base channel width} & \multicolumn{2}{c}{128}   \\
    \textbf{Attention resolutions} & \multicolumn{2}{c}{[32, 16, 8]}  \\
    \textbf{$\#$ Attention heads} & \multicolumn{2}{c}{4}  \\
    \textbf{$\#$ Residual blocks} & \multicolumn{2}{c}{2}  \\
    \textbf{Batch size} & \multicolumn{1}{c}{128} & \multicolumn{1}{c}{32} \\
    \textbf{Learning rate} & \multicolumn{1}{c}{$5e-5$} & \multicolumn{1}{c}{$1e-5$} \\
    \textbf{Channel multipliers} & \multicolumn{1}{c}{[1, 1, 2, 3, 4]} & \multicolumn{1}{c}{[1, 1, 2, 2, 4, 4]} \\
    \textbf{$\#$ Input/Output channels} & \multicolumn{1}{c}{3} & \multicolumn{1}{c}{2} \\
    \bottomrule
    \end{tabular}
    \label{table:ablation_model_details}
\end{table}

\subsection{Comparison Methods for Unconditional Sampling}
\label{appendix:baselines_for_uncond}

We now provide detailed setups for used baselines for unconditional sampling experiments.

\textbf{Oracle diffusion.} For both natural images and multi-coil MRI, we train the oracle diffusion model using clean images without any degradation. The model is trained to predict the noise component of noisy images for both data types.

\textbf{Ambient diffusion~\cite{daras2024ambientdiffusion}.} For natural images, under the same training setup as MSM in noiseless and subsampled data scenario, following the recommendation of \cite{daras2024ambientdiffusion} to apply minimal additional corruption, we define the further degradation operator $\tilde{\vec{S}}$ by dropping one additional $32 \times 32$ pixel box.
The model is trained to directly predict the clean image, which we found to perform better than predicting the noise component.

For MRI data, under the same training setup as the noiseless and subsampled setting of MSM, to define the further degradation operator $\tilde{\vec{S}}$, we drop an additional $10\%$ of the sampling pattern while preserving the autocalibration signal region.
Unlike the natural image case, the model is trained to predict the noise component, which we found to perform better than direct clean image prediction.

During sampling, we use 200 steps of denoising diffusion implicit models (DDIM)~\cite{song2021ddim} and apply the same further degradation configuration used during training to subsample the diffusion iterate.

\textbf{GSURE diffusion~\cite{kawar2023gsure}.} We only apply this method to the RGB face data, not the multi-coil MRI data, because GSURE diffusion's incompatibility with the multi-coil MRI.
Under the same training setup of MSM's noisy and subsampled training setup, we train the GSURE diffusion model to predict the clean images and follow exactly the same training configuration as described in \cite{kawar2023gsure}.

Note that we exclude recent expectation-maximization-based methods~\cite{bai2024emdiff1, bai2025emdiff2}, as their reliance on clean-image initialization is incompatible with our setting, where no clean images are available.

\subsection{Comparison Methods for Imaging Inverse Problems}
\label{appendix:baselines_for_cond}

We now describe the baseline methods used for solving inverse problems.

\textbf{Diffusion posterior sampilng (DPS)~\cite{chung2023dps}.} DPS estimates the gradient of the log-likelihood using the MMSE estimate $\hat{\x}_\theta(\x_t)$ from a pretrained diffusion model as
\begin{equation} \label{eq:dps}
\nabla \log p(\y \mid \x_t) \approx \nabla \log p(\y \mid \hat{\x}_\theta(\x_t)),
\end{equation}
where the gradient is taken with respect to $\x_t$.

Following the original implementation, we set the step size for the likelihood gradient as $\gamma_t = \frac{c}{\|\y-\vec{A}\mathbb{E}[\x_0\mid\x_t]\|_2}$, where $c$ is selected via grid search within the recommended range $[0.1, 10]$.
We used $c=2$ for the super-resolution experiment, $c=0.7$ for box inpainting, and $c=10$ for compressed sensing MRI.

\textbf{Ambient diffusion posterior sampling (A-DPS)~\cite{aali2025ambientposterior}.}  
A-DPS follows the same posterior sampling strategy as DPS but uses a diffusion model trained on noiseless subsampled data.  
The step size is set in the same form as DPS: \(\gamma_t = \frac{c}{\|\y - \vec{A} \mathbb{E}[\x_0 \mid \x_t]\|_2}\), with the constant \(c\) chosen according to each inverse problem:  
\(c = 2\) for super-resolution, \(c = 0.7\) for box inpainting, and \(c = 10\) for compressed sensing MRI.

\textbf{Robust self-supervision via data undersampling (Robust SSDU)~\cite{millard2024robustssdu}.} 
Robust SSDU is designed to handle noisy, subsampled measurements by introducing additional subsampling and added noise to the observed measurements during training.
Specifically, given the noisy subsampled input $\s = \vec{S}_1\z + \n$, where $\n$ is additive Gaussian noise with standard deviation $\sigma_n$, Robust SSDU forms a further corrupted input
\begin{equation}
    \tilde{\s} = \vec{S}_2\s + \tilde{\n},
\end{equation}
where $\vec{S}_2$ has an acceleration factor of $R=2$, and $\tilde{\n}$ is independent Gaussian noise with standard deviation $\sigma_n$, matching the original measurement noise level, following the original implementation.

\textbf{Denoising diffusion null-space model (DDNM)~\cite{wang2022ddnm}.} DDNM also uses a diffusion model trained on clean data and introduces a projection-based update that blends the prior estimate and the measurement:
\begin{equation} \label{eq:ddnm}
    \mathbb{E}[\x_0 \mid \x_t, \y] \approx (\vec{I} - \vec{\Sigma}_t \vec{A}^\dagger \vec{A})\mathbb{E}[\x_0 \mid \x_t] + \vec{\Sigma}_t \vec{A}^\dagger \y,
\end{equation}
where \(\vec{A}^\dagger\) is the pseudoinverse and \(\vec{\Sigma}_t\) is a weighting matrix, such as \(\vec{\Sigma}_t = \lambda_t \vec{I}\) or a spectrally tuned version.

We followed the enhanced version of DDNM described in~\cite[Section 3.3 and Equation (19)]{wang2022ddnm} to specify weight matrix $\vec{\Sigma}_t$ in \eqref{eq:ddnm}.

\subsection{Measuring FID Score}
\label{app:fid}

To compute the Fréchet Inception Distance (FID), we used the implementation provided in the following repository\footnote{\url{https://github.com/mseitzer/pytorch-fid}}.
For each method, we generated 3000 images, then computed FID using features extracted from a pretrained inception network.
For MRI images, which have complex-valued channels, we converted them to magnitude images and replicated the single-channel magnitude three times to form a 3-channel input compatible with the pretrained Inception network.
Note that although the pretrained FID model was trained on natural images, it still reflects perceptual quality on MRI data~\cite{Bendel2023rcGAN_fidfastmri}.

\newpage
\section{Ablation Studies and Discussions}

\subsection{Image Sampling with MSM using Extremely Subsampled Data}
\label{appendix:extreme}

We evaluated the unconditional sampling capability of our framework in a more challenging training scenario with extremely subsampled data. Specifically, we used MRI data with k-space subsampling via random masks at an acceleration factor of $R = 8$, including fully-sampled vertical lines and the central 20 lines for autocalibration.

MSM was configured with a stochastic loop parameter $w=2$ and took 200 sampling steps.
As a baseline, ambient diffusion took 200 sampling steps.
As shown in Table~\ref{table:ablation_mri_fid} and Figure~\ref{fig:ablation_mdiff_extreme}, MSM achieves a lower FID score than the Ambient diffusion on the same data setup, demonstrating the robustness of our approach under high subsampling rates.

\begin{table}[htbp]
    \setstretch{1.5} %
    \centering
    \scriptsize
    \caption{\small FID scores under different training settings on multi-coil brain MR images. \hlgreen{\textbf{Best values}} are highlighted for each training scenario, with comparisons shown when corresponding baseline methods are available. Note how MSM consistently achieves lower FID scores than the Ambient diffusion, even in extremely subsampled data scenarios.}
    \vspace{0.15cm}
    \renewcommand{\arraystretch}{1.5}
    \begin{tabular}{@{}p{2.5cm}p{2.5cm}p{0.5cm}p{0.5cm}@{}p{0.2cm}@{}p{0.5cm}p{0.5cm}@{}}
    \toprule
    \noalign{\vskip -1.67ex}
    \textbf{Training data} & \multicolumn{1}{c}{\textbf{Methods}} & \multicolumn{1}{c}{\textbf{FID}$\downarrow$}  \\ \noalign{\vskip -.65ex}
    \cmidrule{1-3} \noalign{\vskip -1.5ex}
    No degradation & \multicolumn{1}{c}{Oracle diffusion} & \multicolumn{1}{c}{$29.25$}  \\ \cdashline{1-3}  \noalign{\vskip -.64ex}
    \multirow{2}{*}{$R = 4, \eta = 0$} & \multicolumn{1}{c}{MSM} & \multicolumn{1}{c}{\hlgreen{$\bm{43.60}$}}  \\
      & \multicolumn{1}{c}{Ambient diffusion} & \multicolumn{1}{c}{$47.80$}  \\ \cdashline{1-3}  \noalign{\vskip -.64ex}
    \multirow{2}{*}{$R = 8, \eta = 0$} & \multicolumn{1}{c}{MSM} & \multicolumn{1}{c}{\hlgreen{$\bm{74.92}$}}  \\
      & \multicolumn{1}{c}{Ambient diffusion} & \multicolumn{1}{c}{$84.77$}  \\ 
    \bottomrule
    \end{tabular}
    \label{table:ablation_mri_fid}
\end{table}
\begin{figure*}[htbp]
\vspace{-.4cm}
\begin{center}
\includegraphics[width=1.\textwidth]{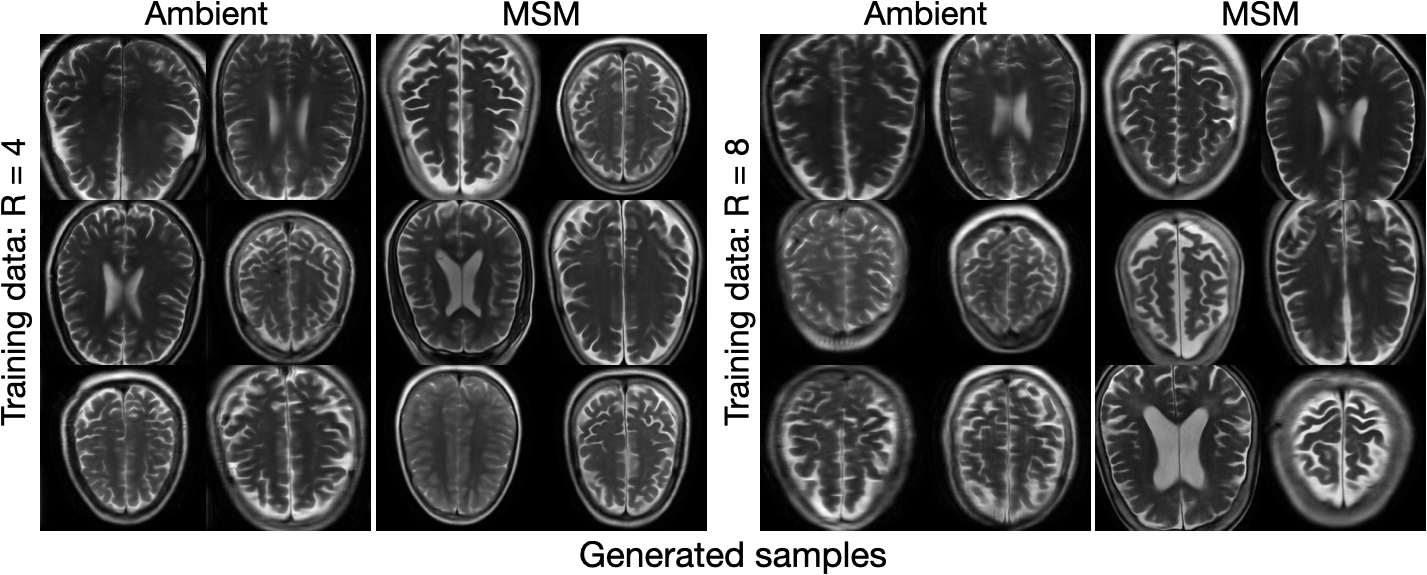}
\end{center}
\caption{Visual comparison of MSM trained under extreme subsampling ($R = 8$) with MSM and baseline methods trained under less degraded conditions.}
\vspace{-.3cm}
\label{fig:ablation_mdiff_extreme}
\end{figure*}

\newpage

\subsection{Solving Inverse Problems with MSM that Trained on Noisy and Subsampled Data}
\label{appendix:solve_inv_wth_noisy_mdiff}

We have shown that MSM trained on subsampled data can effectively solve inverse problems for both natural images and multi-coil MRI.
We further verified that MSM, when trained on noisy and subsampled data, can achieve comparable performance using the same step size for the log-likelihood gradient as in the subsampled-only scenario, as summarized in Table~\ref{table:ablation_quantitative_comparison_invproblems} and Figure~\ref{fig:ablation_mdiff_noisymdiff}.

\begin{table}[htbp]
    \centering
    \scriptsize
    \caption{\small Quantitative results on two natural image inverse problems and two compressed sensing MRI tasks. MSM, trained using only subsampled data, was compared with MSM, trained using noisy and subsampled data. The training scenario is shown in parentheses. Note that MSM trained on noisy and subsampled data achieved comparable performance across all metrics to MSM trained only on subsampled data.}
    \vspace{0.15cm}
    \renewcommand{\arraystretch}{0.9}
    \begin{tabular}{@{}p{1.5cm}p{1.0cm}p{0.5cm}@{}p{0.2cm}@{}p{0.5cm}@{}}
    \toprule
    \textbf{Testing data} &  & \multicolumn{1}{c}{\text{Input}}  & \multicolumn{1}{c}{\text{MSM (Noisy $\&$ Subsampled)}}   & \multicolumn{1}{c}{\text{MSM (Noiseless $\&$ Subsampled)}} \\
    \cmidrule{1-5}\\ \noalign{\vskip -1.9ex}
        \multirow{3}{*}{{\textbf{Box Inpainting}}} & \multicolumn{1}{c}{PSNR$\uparrow$} & \multicolumn{1}{c}{$18.26$} & \multicolumn{1}{c}{$24.16$}  & \multicolumn{1}{c}{$24.71$}  \\[+.95ex]
     & \multicolumn{1}{c}{SSIM$\uparrow$} & \multicolumn{1}{c}{$0.749$} & \multicolumn{1}{c}{$0.864$}   & \multicolumn{1}{c}{$0.867$}   \\[+.75ex]
      & \multicolumn{1}{c}{LPIPS$\downarrow$} & \multicolumn{1}{c}{$0.304$} & \multicolumn{1}{c}{$0.081$}   & \multicolumn{1}{c}{$0.076$}  \\[+.5ex]   
      \cdashline{1-5} \\ \noalign{\vskip -1.ex}
      \multirow{3}{*}{{\textbf{SR} ($\times 4$)}} & \multicolumn{1}{c}{PSNR$\uparrow$} & \multicolumn{1}{c}{$23.21$} & \multicolumn{1}{c}{$27.99$}    & \multicolumn{1}{c}{$28.11$}   \\[+.95ex]
     & \multicolumn{1}{c}{SSIM$\uparrow$} & \multicolumn{1}{c}{$0.728$} & \multicolumn{1}{c}{$0.868$}  & \multicolumn{1}{c}{$0.868$}   \\[+.75ex]
      & \multicolumn{1}{c}{LPIPS$\downarrow$} & \multicolumn{1}{c}{$0.459$} & \multicolumn{1}{c}{$0.127$}  & \multicolumn{1}{c}{$0.117$} \\[+.5ex]
      \bottomrule \\ \noalign{\vskip -1.ex}
    \multirow{3}{*}{\textbf{CS-MRI} ($\times 4$)} & \multicolumn{1}{c}{PSNR$\uparrow$} & \multicolumn{1}{c}{$22.75$} & \multicolumn{1}{c}{$29.74$}  & \multicolumn{1}{c}{$30.71$}   \\[+.95ex]
     & \multicolumn{1}{c}{SSIM$\uparrow$} & \multicolumn{1}{c}{$0.648$}  & \multicolumn{1}{c}{$0.826$}  & \multicolumn{1}{c}{$0.839$}   \\[+.75ex]
      & \multicolumn{1}{c}{LPIPS$\downarrow$} & \multicolumn{1}{c}{$0.306$} & \multicolumn{1}{c}{$0.168$}  & \multicolumn{1}{c}{$0.145$}  \\[+.5ex]
      \cdashline{1-5} \\ \noalign{\vskip -1.ex}
      \multirow{3}{*}{\textbf{CS-MRI} ($\times 6$)} & \multicolumn{1}{c}{PSNR$\uparrow$} & \multicolumn{1}{c}{$21.94$} & \multicolumn{1}{c}{$28.11$}  & \multicolumn{1}{c}{$28.86$}   \\[+.95ex]
     & \multicolumn{1}{c}{SSIM$\uparrow$} & \multicolumn{1}{c}{$0.617$}  & \multicolumn{1}{c}{$0.795$}  & \multicolumn{1}{c}{$0.805$}   \\[+.75ex]
      & \multicolumn{1}{c}{LPIPS$\downarrow$} & \multicolumn{1}{c}{$0.342$} & \multicolumn{1}{c}{$0.192$}  & \multicolumn{1}{c}{$0.168$}  \\[+.5ex]
    \bottomrule
    \end{tabular}
    \label{table:ablation_quantitative_comparison_invproblems}
\end{table}

\begin{figure*}[htbp]
\vspace{-.cm}
\begin{center}
\includegraphics[width=1.0\textwidth]{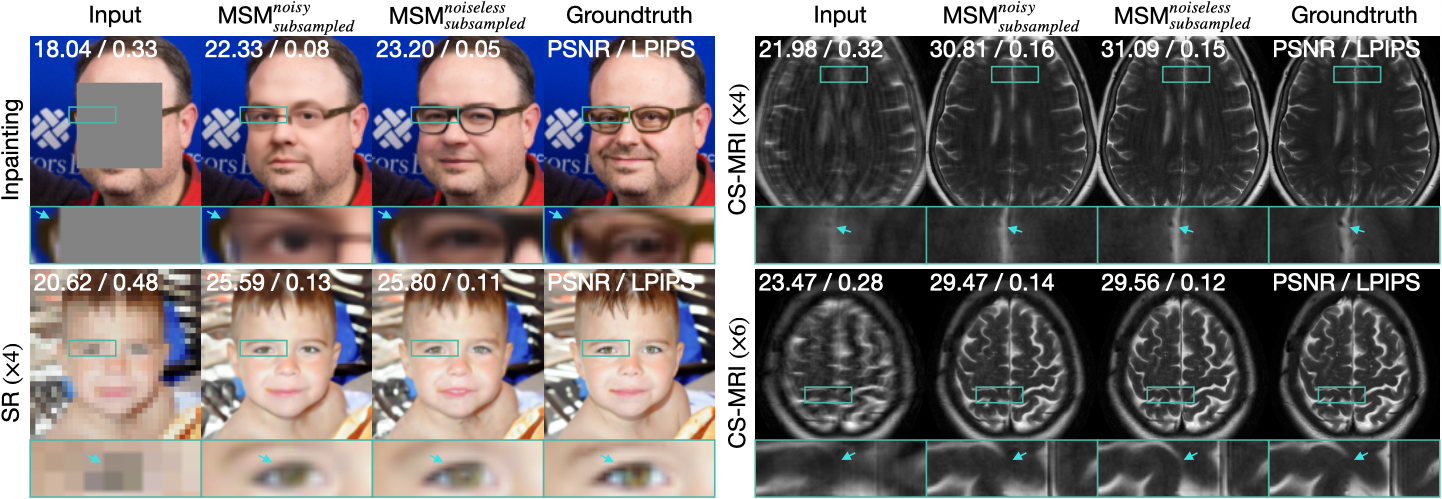}
\end{center}
\caption{Visual comparison between MSM trained on noisy and subsampled data and MSM trained on only subsampled data. Both models produce high-quality results on natural image and MRI tasks.}
\vspace{-.3cm}
\label{fig:ablation_mdiff_noisymdiff}
\end{figure*}

\newpage

\subsection{Comparison with Diffusion-Based Inverse Problem Solvers Trained on Clean Data}
\label{app:comparison_inv_wth_DIS}

We compared MSM against diffusion-based inverse problem solvers—DPS~\cite{chung2023dps} and DDNM~\cite{wang2022ddnm}—that use diffusion priors trained on clean images, on both RGB face images and multi-coil compressed sensing MRI.  
Detailed configurations of these methods are provided in Section~\ref{appendix:baselines_for_cond}.

The measurement noise level was set to \(\eta = 0.01\), and all experimental setups for both MSM and the baselines followed those described in Section~\ref{subsec:rgb_experiment} and Section~\ref{subsec:mri_experiment}.  
The results are summarized in Table~\ref{table:quantitative_comparison_invproblems_with_cleanDIS} and Figure~\ref{fig:ablation_DIS}.

\begin{table}[htbp]
    \centering
    \scriptsize
    \caption{\small Quantitative results on two natural image inverse problems and two compressed sensing MRI tasks. MSM, trained using only subsampled data, is compared with methods that use pretrained diffusion models trained on clean data. The number of iterations used for each method is shown in parentheses. \hlgreen{\textbf{Best}} and \hlblue{second-best} values are highlighted per metric (PSNR, SSIM, LPIPS). Note that despite not having access to the clean data, MSM approaches the performance of the clean data-based methods.}
    \vspace{0.15cm}
    \renewcommand{\arraystretch}{0.9}
    \begin{tabular}{@{}p{1.5cm}p{1.0cm}p{0.5cm}p{0.5cm}@{}p{0.2cm}@{}p{0.5cm}@{}}
    \toprule
    \textbf{Testing data} &  & \multicolumn{1}{c}{\text{Input}}  & \multicolumn{1}{c}{\text{DPS (1000)}}  & \multicolumn{1}{c}{\text{DDNM (200)}}   & \multicolumn{1}{c}{\text{MSM (200)}} \\
    \cmidrule{1-6}\\ \noalign{\vskip -1.9ex}
    \multirow{3}{*}{{\textbf{Box Inpainting}}} & \multicolumn{1}{c}{PSNR$\uparrow$} & \multicolumn{1}{c}{$18.26$}  & \multicolumn{1}{c}{$23.64$}  & \multicolumn{1}{c}{\hlgreen{$\bm{25.16}$}}  & \multicolumn{1}{c}{{\hlblue{$24.93$}}}  \\[+.95ex]
     & \multicolumn{1}{c}{SSIM$\uparrow$} & \multicolumn{1}{c}{$0.749$}  & \multicolumn{1}{c}{$0.864$}  & \multicolumn{1}{c}{\hlgreen{$\bm{0.883}$}}   & \multicolumn{1}{c}{\hlblue{$0.878$}}   \\[+.75ex]
      & \multicolumn{1}{c}{LPIPS$\downarrow$} & \multicolumn{1}{c}{$0.304$}   & \multicolumn{1}{c}{$0.077$} & \multicolumn{1}{c}{\hlblue{$0.071$}}   & \multicolumn{1}{c}{\hlgreen{$\bm{0.066}$}}  \\[+.5ex]   
      \cdashline{1-6} \\ \noalign{\vskip -1.ex}
    \multirow{3}{*}{{\textbf{SR} ($\times 4$)}} & \multicolumn{1}{c}{PSNR$\uparrow$} & \multicolumn{1}{c}{$23.21$}  & \multicolumn{1}{c}{$27.20$} & \multicolumn{1}{c}{\hlgreen{$\bm{28.82}$}}    & \multicolumn{1}{c}{\hlblue{$28.11$}}   \\[+.95ex]
     & \multicolumn{1}{c}{SSIM$\uparrow$} & \multicolumn{1}{c}{$0.728$} & \multicolumn{1}{c}{$0.841$} & \multicolumn{1}{c}{\hlgreen{$\bm{0.897}$}}  & \multicolumn{1}{c}{\hlblue{$0.868$}}   \\[+.75ex]
      & \multicolumn{1}{c}{LPIPS$\downarrow$} & \multicolumn{1}{c}{$0.459$} & \multicolumn{1}{c}{$0.128$} & \multicolumn{1}{c}{\hlblue{$0.126$}}  & \multicolumn{1}{c}{\hlgreen{$\bm{0.117}$}} \\[+.5ex]
      \bottomrule \\ \noalign{\vskip -1.ex}
    \multirow{3}{*}{\textbf{CS-MRI} ($\times 4$)} & \multicolumn{1}{c}{PSNR$\uparrow$} & \multicolumn{1}{c}{$22.75$}   & \multicolumn{1}{c}{\hlblue{$31.31$}} & \multicolumn{1}{c}{\hlgreen{$\bm{32.84}$}}  & \multicolumn{1}{c}{$30.71$}   \\[+.95ex]
     & \multicolumn{1}{c}{SSIM$\uparrow$} & \multicolumn{1}{c}{$0.648$}  & \multicolumn{1}{c}{\hlblue{$0.845$}}  & \multicolumn{1}{c}{\hlgreen{$\bm{0.895}$}}  & \multicolumn{1}{c}{$0.839$}   \\[+.75ex]
      & \multicolumn{1}{c}{LPIPS$\downarrow$} & \multicolumn{1}{c}{$0.306$}  & \multicolumn{1}{c}{\hlblue{$0.112$}}  & \multicolumn{1}{c}{\hlgreen{$\bm{0.104}$}}  & \multicolumn{1}{c}{$0.145$}  \\[+.5ex]
      \cdashline{1-6} \\ \noalign{\vskip -1.ex}
      \multirow{3}{*}{\textbf{CS-MRI} ($\times 6$)} & \multicolumn{1}{c}{PSNR$\uparrow$} & \multicolumn{1}{c}{$21.94$}  & \multicolumn{1}{c}{\hlblue{$29.19$}}  & \multicolumn{1}{c}{\hlgreen{$\bm{29.95}$}}  & \multicolumn{1}{c}{$28.86$}   \\[+.95ex]
     & \multicolumn{1}{c}{SSIM$\uparrow$} & \multicolumn{1}{c}{$0.728$}  & \multicolumn{1}{c}{$0.795$}  & \multicolumn{1}{c}{\hlgreen{$\bm{0.851}$}}  & \multicolumn{1}{c}{\hlblue{$0.805$}}   \\[+.75ex]
      & \multicolumn{1}{c}{LPIPS$\downarrow$} & \multicolumn{1}{c}{$0.459$}  & \multicolumn{1}{c}{\hlblue{$0.149$}} & \multicolumn{1}{c}{\hlgreen{$\bm{0.141}$}}  & \multicolumn{1}{c}{$0.168$}  \\[+.5ex]
    \bottomrule
    \end{tabular}
    \label{table:quantitative_comparison_invproblems_with_cleanDIS}
\end{table}

\begin{figure*}[htbp]
\vspace{-.cm}
\begin{center}
\includegraphics[width=1.0\textwidth]{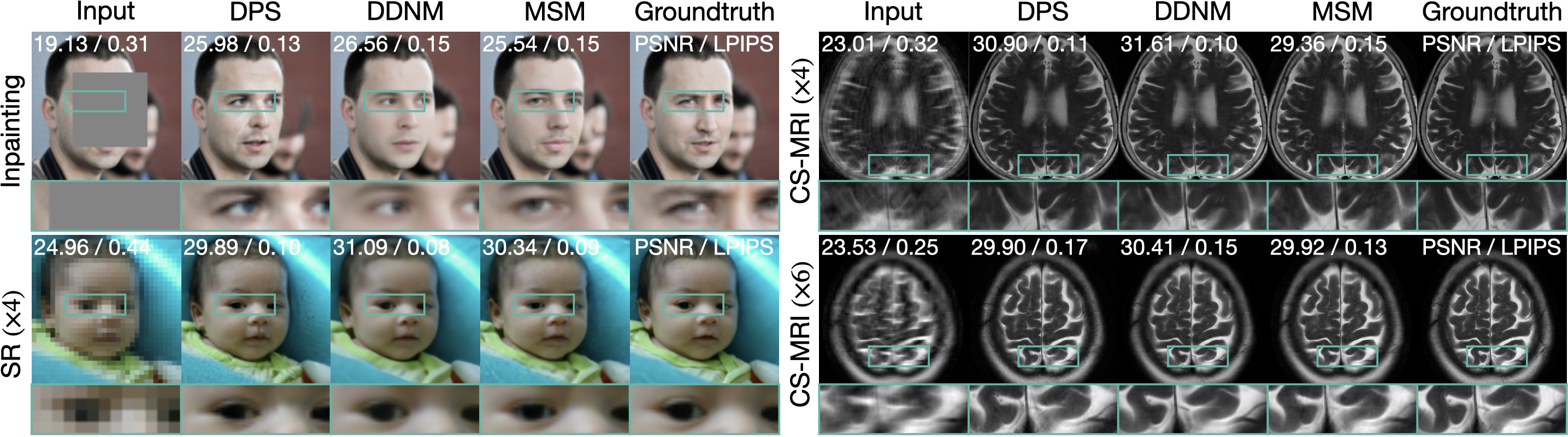}
\end{center}
\caption{Visual comparison of reconstructed images using diffusion-based inverse problem solvers.}
\vspace{-.3cm}
\label{fig:ablation_DIS}
\end{figure*}

\newpage

\subsection{Effect of Stochastic Loop Iterations on Sample Quality}
\label{appendix:effect_of_stochastic_loop}

Our MSM framework includes a configurable parameter \(w\), which controls the number of stochastic iterations used to denoise the full measurement iterates.
As theoretically justified in Section~\ref{sec:theory}, using a larger \(w\) yields a more accurate approximation of the full measurement distribution.

To empirically validate this, we fixed the number of diffusion sampling steps to 10 and explored three different values of \(w\) in \(\{1, 2, 4\}\).
We intentionally fixed the low number of sampling steps to isolate the effect of \(w\), as increasing the number of diffusion steps improves sample quality.
Figure~\ref{fig:ablation_stochasticloop} illustrates that larger \(w\) leads to visually more plausible generations: while \(w=1\) can still produce reasonable samples, it often results in artifacts,
such as visible boundaries in some regions.
As \(w\) increases, the results become more stable and visually coherent.
Table~\ref{table:face_fid_per_b} and Table~\ref{table:fastmri_fid_per_b} further support these observations quantitatively, showing that larger \(w\) achieves better FID scores across both datasets.

\begin{minipage}[htbp]{0.49\textwidth} %
\begin{table}[H]
    \setstretch{1.5} %
    \centering
    \scriptsize
    \caption{\small FID and average time per sampling for MSM with varying stochastic loop iterations $w$ on face images. Note that the trade-off exists where larger \(w\) reduces FID but increases sampling time, and vice versa.}
    \vspace{0.15cm}
    \renewcommand{\arraystretch}{1.1}
    \begin{tabular}{@{}p{1.5cm}p{2.5cm}p{0.5cm}p{0.5cm}@{}p{0.2cm}@{}p{0.5cm}p{0.5cm}p{0.5cm}@{}}
    \toprule
    \noalign{\vskip -.8ex}
    \textbf{Training data} & \multicolumn{1}{c}{\textbf{Methods}} & \multicolumn{1}{c}{\textbf{FID}$\downarrow$} & \multicolumn{1}{c}{\textbf{Time ($s$)}} \\
    \cmidrule{1-4}  \noalign{\vskip -0.8ex}
    \multirow{3}{*}{$p = 0.4, \eta = 0$} & \multicolumn{1}{c}{MSM ($w=4$)} & \multicolumn{1}{c}{$85.02$} & \multicolumn{1}{c}{$1.02$}  \\
      & \multicolumn{1}{c}{MSM ($w=2$)} & \multicolumn{1}{c}{$90.65$} & \multicolumn{1}{c}{$0.51$} \\
      & \multicolumn{1}{c}{MSM ($w=1$)} & \multicolumn{1}{c}{$125.59$} & \multicolumn{1}{c}{$0.27$} \\
    \bottomrule
    \end{tabular}
    \label{table:face_fid_per_b}
\end{table}
\end{minipage}
\hfill %
\begin{minipage}[htbp]{0.49\textwidth} %
\begin{table}[H]
    \setstretch{1.5} %
    \centering
    \scriptsize
    \caption{\small FID and average time per sampling for MSM with varying stochastic loop iterations $w$ on MR images. Similar to face images, increasing \(w\) reduces FID at the cost of longer sampling time.}
    \vspace{0.15cm}
    \renewcommand{\arraystretch}{1.1}
    \begin{tabular}{@{}p{1.5cm}p{2.5cm}p{0.5cm}p{0.5cm}@{}p{0.2cm}@{}p{0.5cm}p{0.5cm}p{0.5cm}@{}}
    \toprule
    \noalign{\vskip -.8ex}
    \textbf{Training data} & \multicolumn{1}{c}{\textbf{Methods}} & \multicolumn{1}{c}{\textbf{FID}$\downarrow$} & \multicolumn{1}{c}{\textbf{Time ($s$)}} \\
    \cmidrule{1-4}  \noalign{\vskip -0.8ex}
    \multirow{3}{*}{$R = 4, \eta = 0$} & \multicolumn{1}{c}{MSM ($w=4$)} & \multicolumn{1}{c}{$75.08$} & \multicolumn{1}{c}{$2.18$}  \\
      & \multicolumn{1}{c}{MSM ($w=2$)} & \multicolumn{1}{c}{$81.97$} & \multicolumn{1}{c}{$1.12$} \\
      & \multicolumn{1}{c}{MSM ($w=1$)} & \multicolumn{1}{c}{$102.89$} & \multicolumn{1}{c}{$0.57$} \\
    \bottomrule
    \end{tabular}
    \label{table:fastmri_fid_per_b}
\end{table}
\small
\end{minipage}

\begin{figure*}[htbp]
\vspace{-.cm}
\begin{center}
\includegraphics[width=1.0\textwidth]{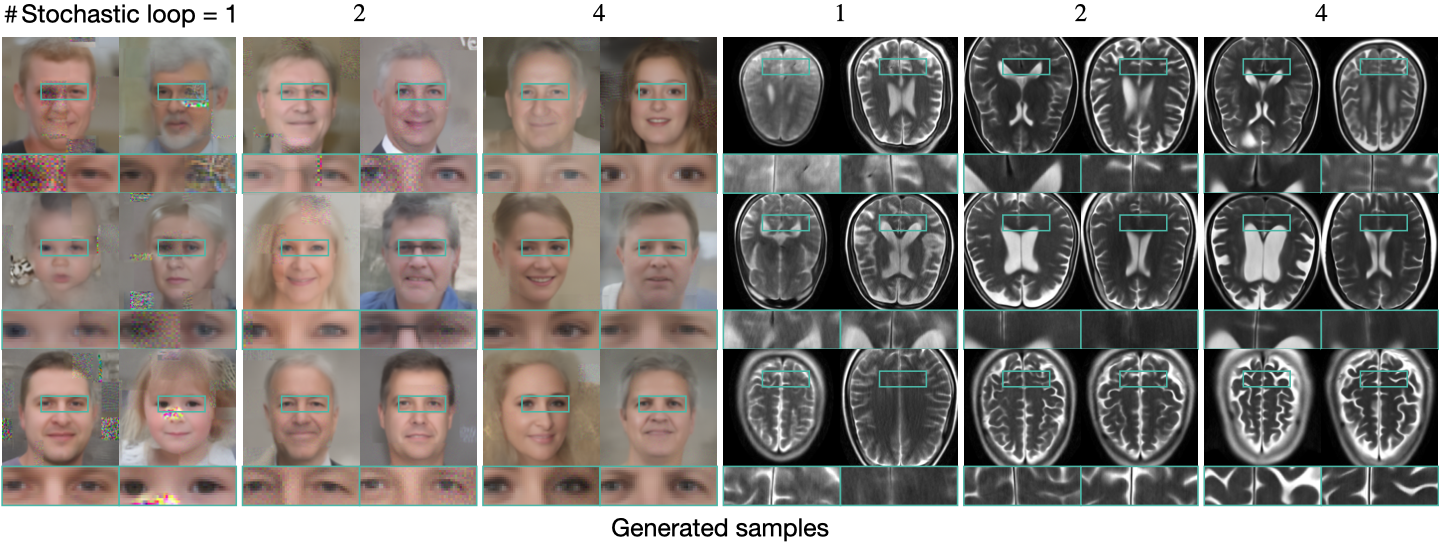}
\end{center}
\caption{Samples generated using different values of the stochastic loop parameter $w$.}
\vspace{-.3cm}
\label{fig:ablation_stochasticloop}
\end{figure*}

\end{document}